\pdfoutput=1
\PassOptionsToPackage{pdftex,
pdfencoding=auto,
pdfnewwindow=true,
pdfusetitle=true,
bookmarks=true,
bookmarksnumbered=true,
bookmarksopen=false,
pdfpagemode=UseThumbs,
bookmarksopenlevel=1,
pdfpagelabels=true,
breaklinks=true,
colorlinks=true,
unicode=true,
pdfborder={0 0 0},
backref=false
}{hyperref}
\PassOptionsToPackage{usenames,svgnames,pdftex}{xcolor}
\documentclass[superscriptaddress,aps,prx,nofootinbib,twocolumn,twoside,reprint,letterpaper,longbibliography,showpacs,10pt,floatfix,tikz,amsmath,amsfonts,amssymb,dvipsnames]{revtex4-2}
\usepackage{xcolor}
\usepackage{amsthm,mathtools,bm,graphicx,verbatim,mathrsfs,units}

\usepackage[utf8]{inputenx}
\usepackage{bbm}

\usepackage{dsfont}
\usepackage{float}
\usepackage{cancel}
\usepackage{booktabs}
\usepackage{mathrsfs}
\usepackage{array}
\usepackage{tabularx}
\usepackage{booktabs}
\usepackage[alwaysadjust,flushleft,pointlessenum]{paralist}
\setdefaultenum{\bfseries1.}{}{}{}

\definecolor{myurlcolor}{rgb}{0,0,0.7}
\definecolor{myrefcolor}{rgb}{0.8,0,0}
\definecolor{purple}{RGB}{128,0,128}
\definecolor{ultramarine}{RGB}{63, 0, 255}
\definecolor{medblue}{RGB}{0, 0, 100}
\definecolor{googleblue}{RGB}{34, 0, 204}
\definecolor{panblue}{RGB}{0,24,150}
\definecolor{carmine}{RGB}{150, 0, 24}
\definecolor{gray}{RGB}{150, 150, 150}
\usepackage{hyperref} 
\hypersetup{
unicode=true,
bookmarksnumbered=false,
bookmarksopen=false,
breaklinks=false,
colorlinks=true,
linkcolor=carmine,
citecolor=googleblue,
urlcolor=panblue,
anchorcolor=OliveGreen}
\usepackage{amsthm}
\usepackage{amsmath}
\usepackage{amssymb}
\usepackage{graphicx}
\usepackage{color}
\usepackage{enumitem} 
\usepackage[normalem]{ulem}
\usepackage{comment} 
\usepackage{tikz}
\usetikzlibrary{calc,arrows.meta, shadows}
\definecolor{gateyellow}{RGB}{252,231,165}
\definecolor{trustedblue}{RGB}{190,222,245}
\definecolor{untrustedred}{RGB}{242,172,168}
\usepackage[caption=false]{subfig}
\usepackage{tikz}
\usetikzlibrary{shapes.geometric,shapes,positioning,shapes.misc}

\usepackage{pifont}
\usepackage{tabularx}
\usepackage{array}

\usepackage[nolist,withpage]{acronym}

\newtheorem{thm}{Theorem}
\newtheorem{theorem}[thm]{Theorem}

\newtheorem{proposition}[thm]{Proposition}

\usepackage{physics}

\begin{document}

\title{Quantum key distribution using generalized contextuality against post-quantum eavesdroppers}
\author{Daniel Centeno}
\email{dcentenodiaz@perimeterinstitute.ca}
\affiliation{Perimeter Institute for Theoretical Physics, Waterloo, Ontario, Canada, N2L 2Y5}
\affiliation{Department of Physics and Astronomy, University of Waterloo, Waterloo, Ontario, Canada, N2L 3G1}
\author{Roberto D. Baldij\~{a}o}
\affiliation{Perimeter Institute for Theoretical Physics, Waterloo, Ontario, Canada, N2L 2Y5}
\author{Maria Ciudad Alañón}
\affiliation{Perimeter Institute for Theoretical Physics, Waterloo, Ontario, Canada, N2L 2Y5}
\affiliation{Department of Physics and Astronomy, University of Waterloo, Waterloo, Ontario, Canada, N2L 3G1}
\author{Yujie Zhang}
\affiliation{Perimeter Institute for Theoretical Physics, Waterloo, Ontario, Canada, N2L 2Y5}
\affiliation{Department of Physics and Astronomy, University of Waterloo, Waterloo, Ontario, Canada, N2L 3G1}
\affiliation{Institute for Quantum Computing, University of Waterloo, Waterloo, Ontario Canada N2L 3G1}
\author{Pedro Lauand}
\affiliation{Perimeter Institute for Theoretical Physics, Waterloo, Ontario, Canada, N2L 2Y5}
\affiliation{Department of Physics and Astronomy, University of Waterloo, Waterloo, Ontario, Canada, N2L 3G1}
\affiliation{Institute for Quantum Computing, University of Waterloo, Waterloo, Ontario Canada N2L 3G1}
\author{Elie Wolfe}
\affiliation{Perimeter Institute for Theoretical Physics, Waterloo, Ontario, Canada, N2L 2Y5}
\affiliation{Department of Physics and Astronomy, University of Waterloo, Waterloo, Ontario, Canada, N2L 3G1}

\begin{abstract}
Bell nonlocality allows secret key to be certified without trusting the internal operation of the devices, and even against eavesdroppers constrained only by the no-signaling principle. We ask whether generalized contextuality, a notion of nonclassicality available in prepare-and-measure experiments on a single system, can play an analogous role. In our approach, Alice and Bob first establish operational equivalences among their local preparation and measurement procedures, which are then treated as theory-independent constraints that any physical intervention on the communication channel, including that of an adversary, must respect. Eve is otherwise unrestricted, that is, her intervention on the transmitted system is not assumed to be described by quantum theory. Within this adversarial model, we show that Eve's optimal guessing probability against individual attacks is the solution of a finite linear program, and hence that a worst-case asymptotic key rate follows from the observed statistics and the operational equivalences alone. We prove that contextuality is necessary for a positive key rate. However, we also show that it is not sufficient. Concretely, the (3,2)-parity-oblivious random access code, known to yield secret key against a quantum eavesdropper, yields none certifiable key rate against a post-quantum one. Conversely, operational equivalences with no analogue among no-signaling constraints can strictly increase the asymptotic key rate. We illustrate this with the prepare-and-measure version of the standard CHSH protocol with an aligned key-generating setting.
\end{abstract}

\maketitle

\section{Introduction}
\label{sec:introduction}

Quantum key distribution (QKD) provides a way of distributing secret information whose security is guaranteed by physical principles rather than by assumptions about the computational power of an eavesdropper. After several decades of theoretical and experimental development, QKD has become one of the most mature applications of quantum information \cite{gisin2002quantum,scarani2009security}. Yet the security of initial QKD protocols relies on a description of the physical devices used by the honest parties. Imperfections or unaccounted degrees of freedom in state preparation, transmission, or measurement can therefore open side channels that are not captured by the idealized security model. Device-independent quantum key distribution (DIQKD) addresses this problem by removing the need to characterize the internal operation of the devices, instead certifying secrecy from observable correlations (plus the assumption of the causal structure) alone \cite{acin2006bell,ekert1991quantum,masanes2011secure}. In this sense, DIQKD provides a particularly strong paradigm for cryptographic security: its security does not rely on trusting a model of the devices used to generate the key. 

A further attractive feature of the device-independent approach is that the relevant constraints need not be restricted to quantum mechanics. Security against an eavesdropper obeying only the no-signaling principle was established early in the development of DIQKD \cite{barrett2005no,scarani2006secrecy,acin2006bell,acin2006efficient}, showing that secret-key generation can remain possible even for an adversary possessing post-quantum physical resources. The essential ingredient in these protocols is the observation of Bell nonlocality, i.e., a sufficiently large violation of a Bell inequality certifies that the observed correlations can be used to establish secrecy~\cite{farkas2021bell}. This raises a broader question: Can we find other operational signatures of nonclassicality that can serve as a basis for secure communication beyond the quantum formalism?

Here we explore generalized contextuality as an operational principle to underlie the security of prepare-and-measure QKD\footnote{Contextuality in the Kochen--Specker sense had previously been considered as a cryptographic resource \cite{liu2025optimal,singh2017quantum,meng2025contextuality,cabello2011hybrid}.}. Generalized contextuality (hereafter also just called contextuality) defines classicality through the existence of a noncontextual ontological model consistent with the operational equivalences between physical procedures  \cite{Spekkens2005}. This applies to arbitrary operational procedures, including preparations, transformations, and measurements, allowing generalized contextuality to describe a broader class of experimental scenarios, such as those involving communication, that lie beyond the scope of Bell nonlocality. Moreover, Bell nonlocality is a special case in that any violation of a Bell inequality witnesses the failure of noncontextuality~\cite{wright2023invertible,baroni2026composable,tavakoli2020measurement, yadavalli2022contextuality}. Contextuality is thus a strictly broader signature of nonclassicality, and one that remains available in prepare-and-measure experiments on a single system, where no Bell inequality applies,  and to scenarios involving transformations and more general quantum processes~\cite{Spekkens2005,Schmid2024structuretheorem, Grabowecky2022, ZhangSchmidYingSpekkens2026}. In addition, contextuality has been used to show quantum advantage in information-theoretic tasks~\cite{schmid2018contextual,spekkens2009preparation,saha2019preparation}.

The first steps toward exploiting contextuality for key distribution have recently been taken \cite{troupe2015contextuality,chaturvedi2021characterising}. Our approach differs from previous contextuality-based QKD proposals in one central respect: we do not assume that the physical system transmitted through the communication channel is governed by quantum mechanics. Instead, in our approach, Alice and Bob first establish a set of operational equivalences (OPEQs) among their local preparation and measurement procedures. This can be done using self-consistent tomography~\cite{MazurekEtAl2021,Grabowecky2022, zhang2026entanglement}, which infers the relevant operational structure directly from the observed statistics without requiring a prior calibrated model of the devices. These equivalences are treated as theory-independent properties of the operational procedures, that is, they are required to hold for any physical test allowed by the operational theory, including tests performed by an adversary. Eve is consequently granted arbitrary control over the communication channel and is not restricted to quantum transformations. The only nonclassical structure retained in the security analysis is therefore the experimentally established set of OPEQs. The adversarial optimization can then be formulated directly over extensions of the observed statistics satisfying the relevant OPEQs, without imposing a quantum state space, quantum channels, or quantum measurements.

A conceptual analogy with DIQKD can be useful in making the assumptions of the two approaches explicit. In the Bell scenario, the causal structure associated with space-like separation motivates the no-signaling constraints. That is, security holds as long as these constraints are valid which in turn rests on the assumed causal structure. In the prepare-and-measure scenario, tomographically complete local laboratories allow Alice and Bob to establish operational equivalences between their respective preparation and measurement procedures. These OPEQs play an analogous role to no-signaling constraints, that is, both constitute theory-independent restrictions on observable statistics. Then, security holds as long as the OPEQs are valid, which in turn rests on the assumed tomographic completeness of the local characterization. This perspective allows us to separate the trusted operational structure from the physical theory governing the communication channel, and hence to ask whether contextuality can certify security against an adversary with resources even beyond those allowed by quantum theory.

Our assumptions also differ from conventional dimension-bounded semi-DI QKD: rather than bounding the Hilbert-space dimension of the communicated system, we assume experimentally characterized operational equivalences among the trusted laboratory procedures~\cite{PawlowskiBrunner2011,WoodheadPironio2015,Woodhead2016}. Beyond strict dimension bounds, other semi-device-independent prepare-and-measure approaches replace the dimension assumption by partial source characterization, for example bounds on the pairwise overlaps of Alice's preparations in receiver-device-independent QKD~\cite{Ioannou2022RDI,Ioannou2022RDIProtocols}, or by physically motivated constraints such as an upper bound on the mean photon number~\cite{VanHimbeeck2017}.

Using our framework, we derive a linear-programming characterization of the adversary's optimal guessing probability. The set of post-quantum extensions consistent with the observed statistics and with the operational equivalences forms a polytope, so the worst-case asymptotic key rate follows from a finite linear program. Among the scenarios we have studied, two are of particular interest: parity-oblivious random access codes (PORACs) and an aligned-key CHSH protocol. We find that the $(3,2)$-PORAC admits no positive key rate against an adversary constrained only by operational equivalences, even though a quantum adversary permits it~\cite{chaturvedi2021characterising}, and that the aligned-key protocol attains a strictly higher rate once all experimentally available equivalences are imposed than when only those mirroring the no-signaling constraints are used (which reproduces exactly the Bell-based rate against a no-signaling adversary~\cite{acin2006efficient}).

This establishes both a strict separation between quantum and post-quantum adversaries, the latter being strictly more powerful, and positive theory-independent key rates certified by operational equivalences that have no counterpart in the Bell scenario. Our results thus demonstrate that generalized contextuality can be used not only as a witness of quantum nonclassicality, but as an operational constraint from which cryptographic security can be analyzed without assuming quantum mechanics for the adversary. More broadly, our results connect the device-independent cryptographic paradigm with the operational framework of generalized contextuality, and identify operational equivalences as a possible alternative foundation for theory-independent cryptography.

The paper is organized as follows. Section~\ref{sec:contextuality} reviews generalized contextuality in the prepare-and-measure setting and introduces the operational equivalences used throughout the paper. Section~\ref{sec:adversary} specifies the adversarial model and the resulting optimization problem when the communication channel is not assumed to be quantum in addition to showing that contextuality is necessary for positive key rate within this model. Section~\ref{sec:results} presents an explicit example, the prepare-and-measure counterpart of the CHSH protocol with an aligned key-generating setting, in which imposing all experimentally available OPEQs yields a strictly higher rate than the one obtained from the no-signaling constraints alone. Section~\ref{sec:experimental} discusses the experimental requirements of the approach, in particular the characterization of the operational equivalences. We conclude in Sec.~\ref{sec:discussion} with a discussion of the assumptions, implications, and limitations of the contextuality-based approach.

\section{Contextuality revisited}
\label{sec:contextuality}

We consider the prepare-and-measure scenario illustrated in Fig.~\ref{fig:prepare_measure}. Alice and Bob operate spatially separated laboratories connected by a physical communication channel. In each round, Alice chooses a preparation procedure $P_x$, labelled by $x\in\mathcal{X}$, and sends the resulting physical system to Bob. Bob independently chooses a measurement $M_y$, with $y\in\mathcal{Y}$, and records an outcome $b\in\mathcal{B}_y$. Operationally, the experiment is specified by the conditional probabilities
\begin{equation}
    p(b|x,y) \equiv p(b|P_x,M_y).
    \label{eq:operational_stats}
\end{equation}
Note that there is no assumption about the physical theory describing the communicated system. In particular, the preparations do not need to be represented by quantum states, nor the measurements by POVMs. The objects $P_x$ and $M_y$ denote laboratory procedures, characterized only through
their observable statistics.

\begin{figure}[t]
\centering
\resizebox{0.7\columnwidth}{!}{%
\begin{tikzpicture}[very thick, line cap=round, line join=round,
  gate/.style={very thick,draw=black!85,
               drop shadow={shadow xshift=0.9mm,shadow yshift=-0.9mm,
                            opacity=0.35,fill=black}},
  trusted/.style={gate,fill=trustedblue},
  labbox/.style={rounded corners=10pt,fill=trustedblue!30,
                 draw=trustedblue!55!black,thick},
  labname/.style={font=\large\sffamily\bfseries,text=trustedblue!45!black}]

  \draw[labbox] (-2.6,-1.3) rectangle (1.55,2.1);
  \node[labname,anchor=north west] at (-2.45,2.0) {Alice};
  \draw[labbox] (2.05,-1.3) rectangle (7.55,2.1);
  \node[labname,anchor=north east] at (7.4,2.0) {Bob};

  \draw[trusted] (1.2, 1.4) -- (1.2, -0.6)
    .. controls (0.2, -0.6) and (-0.8, -0.2) .. (-0.8, 0.4)
    .. controls (-0.8, 1.0) and (0.2, 1.4) .. (1.2, 1.4) -- cycle;
  \draw (-2.2, 0.4) -- (-0.8, 0.4);
  \node[font=\LARGE] at (-1.5, -0.1) {$x$};

  \draw[trusted] (3.6, -0.5) -- (3.6, 1.5)
    .. controls (4.6, 1.5) and (5.6, 1.1) .. (5.6, 0.5)
    .. controls (5.6, -0.1) and (4.6, -0.5) .. (3.6, -0.5) -- cycle;
  \draw (5.6, 0.5) -- (7.15, 0.5);
  \node[font=\LARGE] at (6.7, 1.0) {$b$};

  \draw (1.2, 0.5) -- (3.6, 0.5);
  \draw (1.2, 0.3) -- (3.6, 0.3);

  \draw (2.35, -0.2) -- (3.6, -0.2);
  \node[font=\LARGE] at (2.95, -0.75) {$y$};

\end{tikzpicture}%
}
\caption{Prepare-and-measure scenario. Alice chooses a preparation $P_x$ in
her laboratory and communicates the resulting physical system to Bob, who
chooses a measurement $M_y$ and obtains outcome $b$.}
\label{fig:prepare_measure}
\end{figure}
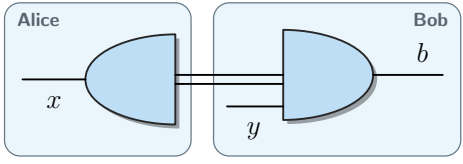

A crucial concept of generalized noncontextuality are \emph{operational equivalences}: relations between laboratory procedures that cannot be distinguished by any experiment available within the operational theory~\cite{Spekkens2005}. For preparations, consider two randomized procedures described by probability distributions $\{\alpha_x\}_x$ and $\{\alpha'_x\}_x$. We write\footnote{Note that for our particular protocols, each preparation $P_x$ is operationally inequivalent to the rest (i.e. $P_x$ is inequivalent to $P_{x'}$ $\forall x\neq x'$). Thus, only convex combinations of them can be operationally equivalent.}
\begin{equation}
    \sum_x \alpha_x P_x
    \simeq
    \sum_x \alpha'_x P_x
    \label{eq:prep_opeq}
\end{equation}
whenever
\begin{equation}
    \sum_x \alpha_x p(b|x,y)
    =
    \sum_x \alpha'_x p(b|x,y)
    \qquad
    \forall\, b,y .
    \label{eq:prep_opeq_operational}
\end{equation}

Thus, no measurement can reveal which of the two randomized preparation procedures was implemented by Alice. An analogous definition applies to Bob's measurement events. Two randomized procedures described by probability distributions $\{\beta_{by}\}_{by}$ and $\{\beta'_{by}\}_{by}$ are operationally equivalent if they occur with the same probability for every allowed preparation. We say
\begin{equation}
   \sum_{b,y} \beta_{by} M^b_y\simeq \sum_{b,y} \beta'_{by} M^b_{y}
    \label{eq:meas_opeq}
\end{equation}
whenever
\begin{equation}
    \sum_{b,y} \beta_{by} p(b|x,y) = \sum_{b,y} \beta'_{by} p(b|x,y) \qquad \forall\, x .
\label{eq:prep_opeq_operational}
\end{equation}
An \emph{ontological model} seeks to explain these operational statistics in terms of an underlying physical state $\lambda\in\Lambda$, referred to as the \emph{ontic state}~\cite{Spekkens2005}; the ontic \emph{space} $\Lambda$ contains all the possible values $\lambda$ could have, and can in principle be continuous. A preparation $P_x$ is represented by a probability distribution $\mu(\lambda|P_x)$ over the ontic space, while a measurement event $[b|M_y]$ is represented by a response function $\xi(b|M_y,\lambda)$. The statistics can then be recovered via 
\begin{equation}
    p(b|x,y) = \int_{\Lambda} d\lambda\, \mu(\lambda|P_x)\, \xi(b|M_y,\lambda),
    \label{eq:ontological_model}
\end{equation}
together with the usual normalization and positivity conditions. The ontic state $\lambda$ can be viewed as a complete ontological specification of the physical system which in the present communication picture, is passed from Alice's laboratory to Bob's. 

The operational equivalences in Eqs.~\eqref{eq:prep_opeq} and \eqref{eq:meas_opeq} are statements about physical procedures and do not by itself constitute an assumption of noncontextuality. Generalized noncontextuality requires that such operational indistinguishability is reflected at the ontological level. In other words, operationally equivalent procedures must be represented identically in any noncontextual ontological explanation.

Therefore, \emph{preparation noncontextuality} means that the equivalence in Eq.~\eqref{eq:prep_opeq} implies
\begin{equation}
    \sum_x \alpha_x\mu(\lambda|P_x) = \sum_x \alpha'_x\mu(\lambda|P_x) \qquad \forall\,\lambda .
    \label{eq:prep_noncontextuality}
\end{equation}
Analogously, \emph{measurement noncontextuality} requires that the equivalence in Eq.~\eqref{eq:meas_opeq} implies
\begin{equation}
    \sum_{b,y}\beta_{by}\xi(b|M_y,\lambda) = \sum_{b,y}\beta'_{by}\xi(b|M_y,\lambda) \qquad \forall\,\lambda .
    \label{eq:meas_noncontextuality}
\end{equation}

We say that the observed prepare-and-measure experiment admits a generalized noncontextual ontological model if there exists a model of the form of Eq.~\eqref{eq:ontological_model} that reproduces the observed statistics $p(b|x,y)$ while respecting all specified operational equivalences, Eqs.~\eqref{eq:prep_noncontextuality} and \eqref{eq:meas_noncontextuality}, at the ontological level. If no such model exists, the experiment is contextual.

\section{Adversarial model}
\label{sec:adversary}

We now specify the adversarial model used in our security analysis. As in standard prepare-and-measure QKD, we regard the communication channel connecting Alice and Bob as completely untrusted. Eve may intercept the physical carrier leaving Alice's laboratory, perform an arbitrary transformation on it, retain arbitrary side information, and forward a physical system to Bob. Crucially, neither Eve's transformation nor the systems available to her are assumed to admit a quantum description. Our security model therefore places no restriction on the physical theory governing the communication channel beyond the operational constraints specified below.

\subsection{Assumptions}

As in any cryptographic protocol, we need some assumptions regarding the isolation of the honest players' laboratories from the adversary. Even in DIQKD, where the internal working of devices is not characterized, the devices cannot be allowed to maliciously communicate private inputs or outputs to Eve~\cite{acin2007device,Primaatmaja2023security}. In this sense, a distinction must be made between \emph{uncharacterized} and \emph{untrusted} device. Device independence removes the need for a microscopic model of the devices, but does not allow arbitrary information leakage from the honest laboratories~\cite{Primaatmaja2023security}. %

Here, we adopt an analogous cryptographic perspective, i.e., Alice's and Bob's laboratories are assumed to be trusted and isolated from Eve, in the sense that no unintended physical degree of freedom carrying private information leaks to the adversary. In particular, their choices ($x$ and $y$) are generated locally and are not available to Eve while she can act on the communicated system. Moreover, any other variable that influences Alice's preparation procedure or Bob's measurement procedure is either included in their local characterization or is independent of Eve. Note that this assumption does not restrict Eve's interaction with the carrier once it has left Alice's laboratory.

Within their trusted laboratories, Alice and Bob use tomographically complete sets of local tests to establish operational equivalences among their preparation and measurement procedures. The resulting operational equivalences are therefore taken to be genuine equivalences of the physical procedures implemented at the boundary of each laboratory. In particular, a preparation equivalence, see Eq.~\eqref{eq:prep_opeq}, asserts that the two preparation mixtures cannot be distinguished by any physically allowed test on the carrier leaving Alice's laboratory. Likewise, a measurement equivalence, see Eq.~\eqref{eq:meas_opeq}, characterizes indistinguishable outcome events of Bob's measurement apparatus for any physical system entering his laboratory. Tomographic completeness is doing two jobs here. First, it extends indistinguishability with respect to the implemented tests to all tests within the operational theory. Second, and crucially for security, it is assumed to hold with respect to all degrees of freedom of the carrier that cross the laboratory boundary, including any that Eve may access. That is, procedures that are indistinguishable according to the complete operational characterization of the honest laboratories cannot become distinguishable under any finer description available to Eve.

An important consequence of our assumptions is that the OPEQs must hold before averaging over any variable correlated with Eve. In particular, an uncharacterized variable influencing Alice's or Bob's local procedures cannot be hidden by averaging over its values, since Eve could condition her attack on that variable and distinguish procedures that appear operationally equivalent to the honest parties. Public or shared randomness is allowed provided it is included in the operational description and the relevant OPEQs hold conditioned on its value. We illustrate this distinction explicitly in Appendix~\ref{app:hidden_randomness}.

Therefore, once an operational equivalence has been established within the laboratory, it must be respected by any future physical intervention. Indeed, if $P\simeq P'$ but an allowed transformation $\mathcal{T}$ produced distinguishable procedures, $\mathcal{T}(P)\not\simeq\mathcal{T}(P')$, then composing $\mathcal{T}$ with a test that distinguishes its outputs would provide a test distinguishing $P$ from $P'$ in the first place. Hence,

\begin{equation}
    P\simeq P' \quad\Longrightarrow\quad \mathcal{T}(P)\simeq\mathcal{T}(P')
    \label{eq:opeq_preservation}
\end{equation}
for any physically allowed transformation $\mathcal{T}$ (note that  $P,P',\mathcal{T}$ are not necessarily quantum procedures). Notice that an analogous reasoning follows for the measurement operational equivalences. Thus, the eavesdropper is constrained by the OPEQs, rather than by quantum theory.

\subsection{Mathematical formulation}
\label{sec:adversarial_math}

We first describe how the raw key\footnote{The \emph{raw key} is the string of correlated symbols retained by Alice and Bob at this stage, before error correction and privacy amplification.} is generated. In each round, Alice chooses a preparation $P_x$ and Bob chooses independently a measurement $M_y$, obtaining an outcome $b$. After the measurement stage is completed, Bob publicly announces his setting $y$. For those pairs $(x,y)$ designated as key-generating rounds, Alice assigns the raw key value $K=\kappa(x,y)$, where the map $\kappa$ is specified by the particular protocol. Bob's outcome $b$ provides his corresponding estimate of the raw key value.

We now specify Eve's attack on this protocol. In each round, she may intercept the system leaving Alice's laboratory, apply an arbitrary transformation, retain an arbitrary nonclassical system, and forward a physical system to Bob. Since this intervention occurs before $y$ is announced, it must be independent of Bob's setting. Nevertheless, Eve may  postpone her final readout until after learning $y$, and choose the measurement on her retained system accordingly (see Fig.~\ref{fig:adversarial_model_2}). To distinguish these two stages explicitly, we denote by $y'$ the choice of Eve's final readout. The complete single-round behaviour is then described by
\begin{equation}
    p(b,e|x,y,y'),
    \label{eq:full_extension}
\end{equation}
where $e$ is the resulting outcome to guess Alice's raw key value $\kappa(x,y)$. Applied independently in each round, this defines the individual-attack model considered throughout this work (See Fig.~\ref{fig:adversarial_model_3}).

\begin{figure}[h]
    \centering
    \resizebox{1\columnwidth}{!}{%
    \begin{tikzpicture}[>=Latex,
  wire/.style={very thick},
  bundle/.style={very thick,double,double distance=2.4pt},
  gate/.style={very thick,draw=black!85,
               drop shadow={shadow xshift=0.9mm,shadow yshift=-0.9mm,
                            opacity=0.35,fill=black}},
  trusted/.style={gate,fill=trustedblue},
  untrusted/.style={gate,fill=untrustedred},
  lab/.style={font=\LARGE}]

\draw[trusted] (0,0.85) -- (-0.6,0.85)
      .. controls (-1.6,0.85) and (-1.6,-0.85) .. (-0.6,-0.85)
      -- (0,-0.85) -- cycle;
\draw[wire] (-3.1,0) -- (-1.35,0);
\node[lab] at (-2.8,-0.5) {$x$};

\draw[bundle] (0.02,0) -- (1.8,0);

\draw[untrusted,rounded corners=6pt] (1.8,-1.9) rectangle (4.4,1.9);
\node[font=\Huge] at (3.1,0) {$\mathcal{T}$};

\draw[bundle] (4.42, 1.45) -- (7.3, 1.45);
\draw[bundle] (4.42,-1.45) -- (7.3,-1.45);

\begin{scope}[shift={(7.3,1.25)}]
  \draw[trusted] (0,1.1) -- (0.55,1.1)
        .. controls (1.6,1.1) and (1.6,-1.1) .. (0.55,-1.1)
        -- (0,-1.1) -- cycle;
  \draw[wire] (1.32,0) -- (3,0);
\end{scope}
\node[lab] at (10,1.7) {$b$};

\begin{scope}[shift={(7.3,-1.25)}]
  \draw[untrusted] (0,1.1) -- (0.55,1.1)
        .. controls (1.6,1.1) and (1.6,-1.1) .. (0.55,-1.1)
        -- (0,-1.1) -- cycle;
  \draw[wire] (1.32,0) -- (3,0);
\end{scope}
\node[lab] at (10,-1.7) {$e$};

\draw[wire] (4.7,2.15) -- (7.3,2.15);
\draw[wire] (5.8,2.15) -- (5.8,-0.55) -- (7.3,-0.55);
\fill (5.8,2.15) circle (2.2pt);
\node[lab] at (4.85,2.55) {$y$};

\end{tikzpicture}
    }
    \caption{Adversarial prepare-and-measure scenario allowing Eve to have access to $y$ once the measurement of Bob is performed. The blue devices represent the trusted procedures and the red ones, the untrusted ones.}
    \label{fig:adversarial_model_2}
\end{figure}
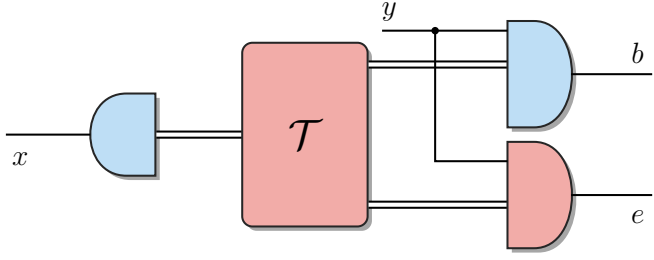

The aim of the optimization of the adversarial attack is to bound the amount of secret key that can be extracted from the observed correlations. Denoting by $E$ Eve's classical readout, per Ref.~\cite{csiszar1978broadcast}, the asymptotic key rate is lower bounded by  %
\begin{equation}
    r \geq H(K|E,Y)-H(K|B,Y).
    \label{eq:rate}
\end{equation}

For the binary raw keys considered here, Eve's conditional Shannon entropy can be lower bounded directly in terms of her optimal guessing probability\footnote{The guessing probability $P_{\mathrm{guess}}(K|E,Y)$ is the probability that Eve correctly identifies the raw key value using her classical readout and the publicly announced setting $y$, maximized over all her admissible strategies.}. The binary reverse-Fano inequality\footnote{Note that for non-binary raw key alphabets, an analogous reverse-Fano relation lower bounds $H(K|E,Y)$ in terms of Eve's guessing probability. However, the expression used here is the specialization to a binary key.} \cite{feder1994relations,rioul2023interplay} gives
\begin{equation}
    H(K|E,Y)
    \geq
    2\left[
        1-P_{\mathrm{guess}}(K|E,Y)
    \right].
    \label{eq:reverse_fano}
\end{equation}
Combining Eqs.~\eqref{eq:rate} and~\eqref{eq:reverse_fano}, we obtain
\begin{equation}
    r
    \geq
    2\left[
        1-P_{\mathrm{guess}}(K|E,Y)
    \right]
    -
    H(K|B,Y).
    \label{eq:key_rate_guess}
\end{equation}
Thus, for fixed observed statistics, obtaining a worst-case lower bound on the key rate reduces to maximizing Eve's guessing probability over all individual attacks compatible with our operational assumptions.\footnote{The same guessing probability defines the conditional min-entropy of Eve's classical readout, $H_{\min}(K|E,Y)=-\log_2P_{\mathrm{guess}}(K|E,Y)$. For a binary raw key this also provides a lower bound on the corresponding Shannon conditional entropy, but it is generally weaker than Eq.~\eqref{eq:reverse_fano}. Smooth min-entropy becomes particularly relevant for composable privacy amplification against general multi-round attacks \cite{Tomamichel2011fbe,tomamichel2012tight} which go beyond the scope of this work.} %

\begin{figure}[h]
    \centering
    \resizebox{1\columnwidth}{!}{%
    \begin{tikzpicture}[>=Latex,
  wire/.style={very thick},
  bundle/.style={very thick,double,double distance=2.4pt},
  gate/.style={very thick,draw=black!85,
               drop shadow={shadow xshift=0.9mm,shadow yshift=-0.9mm,
                            opacity=0.35,fill=black}},
  trusted/.style={gate,fill=trustedblue},
  untrusted/.style={gate,fill=untrustedred},
  lab/.style={font=\LARGE}]

\draw[trusted] (0,0.85) -- (-0.6,0.85)
      .. controls (-1.6,0.85) and (-1.6,-0.85) .. (-0.6,-0.85)
      -- (0,-0.85) -- cycle;
\draw[wire] (-3.1,0) -- (-1.35,0);
\node[lab] at (-2.8,-0.5) {$x$};

\draw[bundle] (0.02,0) -- (1.8,0);

\draw[untrusted,rounded corners=6pt] (1.8,-1.9) rectangle (4.4,1.9);
\node[font=\Huge] at (3.1,0) {$\mathcal{T}$};

\draw[bundle] (4.42, 1.45) -- (7.3, 1.45);
\draw[bundle] (4.42,-1.45) -- (7.3,-1.45);

\begin{scope}[shift={(7.3,1.25)}]
  \draw[trusted] (0,1.1) -- (0.55,1.1)
        .. controls (1.6,1.1) and (1.6,-1.1) .. (0.55,-1.1)
        -- (0,-1.1) -- cycle;
  \draw[wire] (1.32,0) -- (3,0);
\end{scope}
\node[lab] at (10,1.7) {$b$};

\begin{scope}[shift={(7.3,-1.25)}]
  \draw[untrusted] (0,1.1) -- (0.55,1.1)
        .. controls (1.6,1.1) and (1.6,-1.1) .. (0.55,-1.1)
        -- (0,-1.1) -- cycle;
  \draw[wire] (1.32,0) -- (3,0);
\end{scope}
\node[lab] at (10,-1.7) {$e$};

\draw[wire] (4.7, 2.15) -- (7.3, 2.15);
\node[lab] at (4.85, 2.55) {$y$};

\draw[wire] (4.7,-2.15) -- (7.3,-2.15);
\node[lab] at (4.85,-2.55) {$y'$};

\end{tikzpicture}
    }
    \caption{Adversarial prepare-and-measure scenario in which Bob and Eve have independent settings $y$ and $y'$ used in the optimization. The blue devices represent the trusted procedures and the red ones, the untrusted ones.}
    \label{fig:adversarial_model_3}
\end{figure}
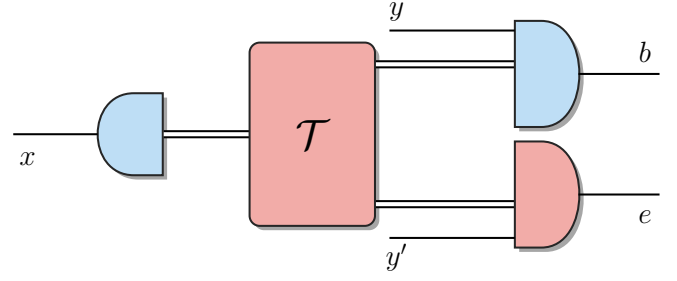

Note that not all rounds of the experiment need to be used for key generation. We denote by $\mathcal{Y}_{\rm key}$ the set of Bob's settings that may contribute to the raw key, and by $\mathcal{X}_{\rm key}^y$ the preparations retained when Bob announces $y$, with corresponding probabilities $q(y)$ and $q(x|y)$. Moreover, since in the actual protocol Eve chooses her setting $y'$ after learning $y$, the guessing probability only includes the rounds for which $y'=y$. Hence, 
\begin{equation}
\begin{aligned}
    P_{\mathrm{guess}}(K|E,Y) =\\
    \sum_{\substack{y\in\mathcal{Y}_{\rm key}\\
            x\in\mathcal{X}_{\rm key}^{y}\\
            b}}q(y)q(x|y) p\!\left(b,e=\kappa(x,y)\middle|x,y,y'=y\right)
    \end{aligned}
\end{equation}
Importantly,  The Alice-Bob statistics must be recovered upon marginalizing over Eve,
\begin{equation}
    \sum_e p(b,e|x,y,y')=p_{\mathrm{obs}}(b|x,y)\quad \forall y'.
    \label{eq:AB_marginal}
\end{equation} 
Given our assumptions, the operational equivalences impose linear constraints directly on the extended distribution $p(b,e|x,y,y')$. For every preparation equivalence of the form of Eq.~\eqref{eq:prep_opeq}, indistinguishability must persist for any test performed by Eve and Bob (recall Eq.~\eqref{eq:opeq_preservation}), yielding
\begin{equation}
    \sum_x \alpha_x p(b,e|x,y,y')=\sum_x \alpha'_x p(b,e|x,y,y')\qquad\forall\,b,e,y,y'.
    \label{eq:eve_prep_opeq}
\end{equation}
Similarly, for every measurement operational equivalence of the form of
Eq.~\eqref{eq:meas_opeq}, we require
\begin{equation}
\sum_{b,y}\beta_{by}p(b,e|x,y,y')=\sum_{b,y}\beta'_{by}p(b,e|x,y,y')\qquad\forall\,x,e,y'.
    \label{eq:eve_meas_opeq}
\end{equation}

Finally, the causal constraints from the model are the no-signaling conditions from $y$ to $e$,
\begin{equation}
     \sum_b p(b,e|x,y,y')=\sum_b p(b,e|x,\widetilde y,y'),\qquad \forall\,e,x,y,\widetilde y,y',
    \label{eq:lp_ns_bob_to_eve}
\end{equation}
and from $y'$ to $b$ which is already imposed by the recovery of the observed marginals, Eq.~\eqref{eq:AB_marginal}.

Therefore, the optimization of the adversarial attack is

\begin{align}
    \underset{p}{\mathrm{max}}
    \quad &  P_{\mathrm{guess}}(K|E,Y)\quad \mathrm{such} \; \mathrm{that}
    \nonumber\\
    & p(b,e|x,y,y')\geq0 ,
    \nonumber\\
    & \text{Eq.~(\ref{eq:AB_marginal})},
    \label{eq:adversarial_LP}\\
    & \text{preparation OPEQs,
      Eq.~\eqref{eq:eve_prep_opeq}},
    \nonumber\\
    & \text{measurement OPEQs,
      Eq.~\eqref{eq:eve_meas_opeq}},
    \nonumber\\
    & \text{no-signaling from $y$ to $e$,
      Eq.~\eqref{eq:lp_ns_bob_to_eve}}.
    \nonumber
\end{align}

Although Eq.~\eqref{eq:adversarial_LP} gives the complete optimization, it need not be solved as a single large linear program. For different values of $y'$, the constraints above separate into independent blocks, while the objective contains only the diagonal terms $y'=y$. Consequently, the optimum can be obtained more efficiently by solving one fixed-readout linear program for each $y\in\mathcal{Y}_{\rm key}$ and combining the resulting guessing probabilities with the weights $q(y)$. We prove this decomposition in Appendix~\ref{app:fixed_setting}.

\subsection{Contextuality as a necessary condition}
\label{sec:theorem}
A first general consequence of the optimization in Eq.~\eqref{eq:adversarial_LP} is that a positive key rate requires the observed statistics to be contextual. That is, whenever a noncontextual ontological model exists, it can be turned into an attack under which Alice and Bob cannot distill any secret key. %

\begin{theorem}
Consider a prepare-and-measure experiment specified by the probabilities $p_{\rm obs}(b|x,y)$ together with a set of preparation and measurement OPEQs. If this experiment admits a generalized noncontextual ontological model, then the asymptotic secret key rate vanishes under our adversarial model.
\label{theorem}
\end{theorem}

The proof is given in Appendix~\ref{app:ncontextual_zero_key}. Note that, as we will see in the following section, the converse does not hold, that is, contextuality is necessary but not sufficient for key generation.

\section{Example}
\label{sec:results}
The adversarial framework introduced above applies to arbitrary finite prepare-and-measure experiments once the preparations, measurements and the operational equivalences are specified. Given such specification, the linear program above provides the value of $P_{\rm guess}$, which in turn can be used to determine whether Alice and Bob can ever extract a secure key. We have implemented the corresponding linear programs in an open-source numerical package \cite{ContextualityQKDGitHub}, which allows the adversarial guessing probability and the resulting key-rate bound to be evaluated directly from this operational description.

Rather than presenting an exhaustive catalog of these examples here, we focus on an insightful one and refer to Appendix~\ref{app:keyrun_tables} for others that we have explored. The latter includes a variety of prepare-and-measure scenarios such as the prepare-and-measure CHSH, the hexagon scenario~\cite{mazurek2016experimental} and several Kochen-Specker constructions~\cite{cabello1996bell,peres1991two}. Moreover, exploring parity-oblivious random access codes (PORACs)~\cite{spekkens2009preparation} demonstrates that the adversarial model considered here is strictly stronger than restricting the eavesdropper to quantum mechanics. This is because in Ref.~\cite{chaturvedi2021characterising}, they obtain a positive key rate against a quantum eavesdropper, while we obtain no key rate according to our model (see Appendix~\ref{app:PORAC} for more details).

The example we consider here shows that operational equivalences beyond those that mimic the no-signaling constraints can lead directly to stronger cryptographic bounds. We consider the prepare-and-measure version of the protocol of Ref.~\cite{acin2006efficient}, in which the CHSH test is supplemented by an additional setting of Alice's aligned with one of Bob's, used for key generation. More concretely, Bob performs two binary measurements $M_0$ and $M_1$, with effects $M^b_y$ for outcome $b\in\{0,1\}$. Alice has six preparations, that is, the four preparations, $P_0,P_1,P_2,P_3$, used in the CHSH test, and two further preparations $P_4$ and $P_5$ aligned with the two effects of Bob's measurement $M_0$. The preparations and effects are depicted in Fig.~\ref{fig:CHSH_plus}. In the qubit realization, $P_0=\ket{0}$, $P_1=\ket{1}$, $P_2=\ket{+}$ and $P_3=\ket{-}$, Bob's measurements $M_0$ and $M_1$ are along the directions $(Z+X)/\sqrt{2}$ and $(Z-X)/\sqrt{2}$, respectively, and $P_4$ and $P_5$ are the two antipodal states along $(Z+X)/\sqrt{2}$.

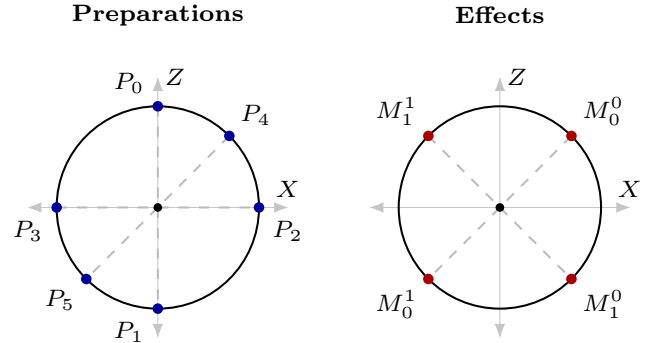
\begin{figure}[h]
    \centering
    \resizebox{1\columnwidth}{!}{%
    \begin{tikzpicture}[
  >={Latex[length=1.4mm]},
  st/.style={circle,fill=blue!60!black,inner sep=1.1pt},
  ef/.style={circle,fill=red!65!black,inner sep=1.1pt},
  ctr/.style={circle,fill=black,inner sep=0.9pt},
  ax/.style={gray!55,very thin,<->},
  diam/.style={dashed,gray!65,semithick},
  lbl/.style={font=\scriptsize},
  ttl/.style={font=\scriptsize\bfseries},
]
\def\R{1.05}   %
\def\D{3.55}   %

\begin{scope}
  \draw[ax] (-\R-0.3,0) -- (\R+0.3,0)
    node[above=2pt,lbl,black,inner sep=1pt]{$X$};
  \draw[ax] (0,-\R-0.3) -- (0,\R+0.3)
    node[right=1pt,lbl,black,inner sep=1pt]{$Z$};
  \draw[semithick] (0,0) circle (\R);
  \draw[diam] (0,\R) -- (0,-\R);
  \draw[diam] (-\R,0) -- (\R,0);
  \draw[diam] (45:\R) -- (225:\R);
  \node[st] (s0) at (0,\R) {};    \node[st] (s1) at (0,-\R) {};
  \node[st] (s2) at (\R,0) {};    \node[st] (s3) at (-\R,0) {};
  \node[st] (s4) at (45:\R) {};   \node[st] (s5) at (225:\R) {};
  \node[ctr] at (0,0) {};
  \node[lbl,above left=1pt]  at (s0) {$P_0$};
  \node[lbl,below left=1pt]  at (s1) {$P_1$};
  \node[lbl,below right=1pt] at (s2) {$P_2$};
  \node[lbl,below left=1pt]  at (s3) {$P_3$};
  \node[lbl,above right=0pt] at (s4) {$P_4$};
  \node[lbl,below left=0pt]  at (s5) {$P_5$};
  \node[ttl] at (0,\R+0.95) {Preparations};
\end{scope}

\begin{scope}[xshift=\D cm]
  \draw[ax] (-\R-0.3,0) -- (\R+0.3,0)
    node[above=2pt,lbl,black,inner sep=1pt]{$X$};
  \draw[ax] (0,-\R-0.3) -- (0,\R+0.3)
    node[right=1pt,lbl,black,inner sep=1pt]{$Z$};
  \draw[semithick] (0,0) circle (\R);
  \draw[diam] (45:\R) -- (225:\R);
  \draw[diam] (135:\R) -- (315:\R);
  \node[ef] (e00) at (45:\R)  {};  \node[ef] (e11) at (135:\R) {};
  \node[ef] (e10) at (225:\R) {};  \node[ef] (e01) at (315:\R) {};
  \node[ctr] at (0,0) {};
  \node[lbl,above right=0pt] at (e00) {$M^0_0$};
  \node[lbl,above left=0pt]  at (e11) {$M^1_1$};
  \node[lbl,below left=0pt]  at (e10) {$M^1_0$};
  \node[lbl,below right=0pt] at (e01) {$M^0_1$};
  \node[ttl] at (0,\R+0.95) {Effects};
\end{scope}
\end{tikzpicture}
    }
    \caption{Preparations and effects of the aligned-key protocol in the qubit realization, drawn in the $X$--$Z$ plane: $P_0=\ket{0}$, $P_1=\ket{1}$, $P_2=\ket{+}$, $P_3=\ket{-}$, and $P_4$, $P_5$ the two antipodal states along $(Z+X)/\sqrt{2}$, aligned with the effects $M^0_0$ and $M^1_0$ of Bob's key-generating measurement $M_0$. Pairs of antipodal points, joined by dashed diameters, correspond to bases.}
    \label{fig:CHSH_plus}
\end{figure}

The six preparations satisfy
\begin{equation}
    \frac{1}{2}(P_0+P_1) \simeq \frac{1}{2}(P_2+P_3) \simeq \frac{1}{2}(P_4+P_5),
    \label{eq:aligned_chsh_ns_prep}
\end{equation}
and Bob's measurements satisfy
\begin{equation}
    \frac{1}{2}\left(M^0_0+M^1_0\right) \simeq \frac{1}{2}\left(M^0_1+M^1_1\right).
    \label{eq:chsh_pm_opeq2}
\end{equation}
These two sets of operational equivalences play exactly the role of the no-signaling constraints from Alice to Bob and from Bob to Alice, respectively.

The six preparations, however, satisfy an additional operational equivalence that has no counterpart among the no-signaling constraints,
\begin{equation}
    \frac{1}{2}(P_0+P_2) \simeq \frac{2+\sqrt{2}}{4}P_4 + \frac{2-\sqrt{2}}{4}P_5 .
    \label{eq:aligned_chsh_extra}
\end{equation}
For the qubit realization this relation follows directly from the geometry of the preparations in the $X$--$Z$ plane, but in our analysis it is used only as an experimentally established OPEQ.

As mentioned, only the additional preparations are used to generate the key. Thus, we take $\mathcal{X}_{\mathrm{key}}=\{4,5\}$ and $\mathcal{Y}_{\mathrm{key}}=\{0\}$, with raw key being $\kappa(x,0)=x \: mod(4)$. Hence, in the ideal realization, Bob recovers the key bit perfectly from his outcome. The remaining preparations and measurement setting are used only to characterize the correlations that constrain Eve. As in the original Bell protocol~\cite{acin2006efficient}, the key-generating choices can be made with probabilities arbitrarily close to one in the asymptotic limit, while the remaining settings are sampled sufficiently often to estimate the required conditional probabilities. Consequently, the asymptotic key rate approaches the rate obtained from the key-generating rounds. 

This scenario allows us to isolate the cryptographic contribution of the additional OPEQ in Eq.~\eqref{eq:aligned_chsh_extra}. We first perform the adversarial optimization imposing the complete set of experimentally available OPEQs, Eqs.~\eqref{eq:aligned_chsh_ns_prep}, \eqref{eq:chsh_pm_opeq2} and \eqref{eq:aligned_chsh_extra}, and we obtain $r_{\mathrm{all\;OPEQs}} \geq 0.5857$. We then repeat the optimization discarding Eq.~\eqref{eq:aligned_chsh_extra}, that is, imposing only the no-signaling-type OPEQs, Eqs.~\eqref{eq:aligned_chsh_ns_prep} and \eqref{eq:chsh_pm_opeq2}. Under this restriction, the same observed Alice--Bob statistics instead yield $r_{\mathrm{NS}} \geq 0.4142$, reproducing the no-signaling lower bound on the key-rate of Ref.~\cite{acin2006efficient}. This agreement is expected because when only Eqs.~\eqref{eq:aligned_chsh_ns_prep} and \eqref{eq:chsh_pm_opeq2} are imposed, the constraints entering our prepare-and-measure optimization are precisely those corresponding to no-signaling from Alice to Bob and from Bob to Alice in the Bell scenario.

Since Alice and Bob are perfectly correlated on the key-generating settings, the gap between the two lower bounds on the rate arises entirely from the stronger restriction placed on Eve's admissible extensions. The additional operational equivalence reduces her optimal guessing probability even though the observed correlations are unchanged. This provides a direct example in which experimentally established operational equivalences in a prepare-and-measure experiment certify a larger post-quantum lower bound on the key rate than that obtained from the usual Bell-based DIQKD constraints.

\section{Experimental considerations}
\label{sec:experimental}
The protocols considered above can in principle be implemented using essentially the same optical components as conventional polarization-encoded prepare-and-measure QKD~\cite{gisin2002quantum,scarani2009security}. The main experimental distinction from conventional prepare-and-measure QKD lies instead in the role of device characterization. In standard QKD, imperfections of the source and measurement apparatus are typically incorporated into a trusted quantum model of the devices. Here, the security analysis relies directly on experimentally established operational equivalences, and thus, on the condition of tomographic completeness for that characterization. Consequently, each preparation or measurement need not coincide exactly with its nominal quantum description. Rather, the relevant mixtures must be operationally indistinguishable to the precision required by the security analysis.

Experimentally, inferring the operational equivalences is itself a nontrivial characterization problem. Earlier experimental tests of generalized noncontextuality already highlighted both the difficulty of realizing exact operational equivalences and the role of tomographic completeness in establishing them~\cite{mazurek2016experimental}. A natural approach here is to reconstruct the preparation and measurement spaces self-consistently from the observed statistics, without assuming a prior quantum characterization of either side. Such self-consistent GPT tomography has been demonstrated experimentally and can be used both to test hypotheses for the operational dimension and to reconstruct the corresponding state and effect spaces~\cite{MazurekEtAl2021}. Closely related self-consistent reconstruction ideas also underlie the gauge-independent identification of operational identities used in Ref.~\cite{zhang2026entanglement}.

For example, arranging the observed probabilities of a prepare-and-measure experiment into the matrix
\begin{align}
    \mathbf P_{(x),(b,y)}:=p(b|xy),
\end{align}
the preparation and measurement OPEQs correspond, under tomographic completeness, to the appropriate left and right null spaces of $\mathbf P$. This characterization is invariant under the invertible linear gauge transformations inherent in self-consistent tomography~\cite{MazurekEtAl2021,zhang2026entanglement}. In the ideal infinite-data limit, these relations appear as exact null directions. With finite data, however, the empirical matrix $\widehat{\mathbf P}$ will generically have no exact null directions. One must therefore determine the effective operational dimension and distinguish genuine linear dependencies from small singular values generated by statistical fluctuations. The reliability of an inferred OPEQ thus depends not only on the number of experimental samples, but also on the conditioning of the tomographic reconstruction and on the spectral gap separating the approximately null subspace from the remaining
operational directions.

There is also a distinction between inferring an operational identity in the reconstructed model and realizing it exactly with laboratory procedures. Exact operational equivalence cannot generally be expected at finite experimental precision. In the photonic experiment of Ref.~\cite{mazurek2016experimental}, this difficulty was addressed using secondary procedures, constructed as convex mixtures of experimentally implemented procedures so as to satisfy the desired operational equivalences exactly. Alternatively, one can formulate tests that are robust to imperfect operational equivalences~\cite{Pusey2018}. Similar techniques could provide a natural starting point for a robust experimental implementation of the present protocol.

Since the operational equivalences are characterized locally, Alice's by testing her preparations against a tomographically complete set of measurements inside her laboratory and Bob's likewise inside his, the settings used for the QKD protocol need not themselves form a tomographically complete set. Auxiliary local preparations and measurements may be introduced solely for characterization. For a polarization implementation, for example, additional measurements outside the $X$--$Z$ plane can be used to complete the tomographic set  without entering either the key-generation or contextuality-test rounds.  Experimentally, it would also be preferable to include such characterization rounds throughout the QKD run rather than relying only on a pre-calibration, thereby testing the stability of the operational equivalences during data acquisition. This is particularly relevant here because slow device drift may change not only the observed error rate, as in conventional QKD, but also the operational relations that determine Eve's admissible set. As in efficient prepare-and-measure QKD~\cite{LoChauArdehali2005}, the key-generating settings can be strongly favored while the remaining settings are sampled only often enough to estimate the required operational structure.

The potential advantage of the present approach should therefore be understood primarily as a change in the assumptions underlying the security analysis, although the richer set of experimentally accessible OPEQs can also strengthen the resulting key-rate bound in some scenarios. In the present work, the relevant operational equivalences in the QKD protocol are assumed to be ideally characterized. A realistic implementation will of course verify these operational equivalences only approximately. As in many realistic QKD protocols~\cite{Xu2020}, incorporating finite statistics, device drift, loss, and other imperfections into robust OPEQ constraints is an important next step, which we leave to future experimental work.

\section{Discussion}
\label{sec:discussion}

Our results show that generalized contextuality can provide an operational basis for prepare-and-measure QKD without assuming that the communicated system, or Eve's intervention on it, is governed by quantum mechanics. The central ingredient is the set of operational equivalences established within the trusted laboratories of Alice and Bob. These relations directly constrain the admissible extensions of the observed statistics and hence Eve’s guessing probability. In scenarios that are the prepare-and-measure counterpart of Bell experiments, the OPEQs that mirror no-signaling constraints recover the corresponding post-quantum lower bounds on the key rates. Prepare-and-measure experiments may, however, contain operational equivalences with no no-signaling counterpart, and exploiting them can strictly increase the certified rate for the same observed statistics.

At the same time, the presence of contextuality does not by itself guarantee secrecy. We showed that no positive key rate can be extracted from statistics admitting a noncontextual model, and the $(3,2)$-PORAC example shows that the converse fails: correlations that are contextual, and that yield a positive key rate against a quantum adversary, may nevertheless admit a post-quantum extension compatible with all their OPEQs for which no positive rate can be certified. Contextuality is therefore necessary but not sufficient for secret-key generation in our adversarial model, much as Bell nonlocality is known to be insufficient for broad classes of device-independent protocols~\cite{farkas2021bell}.

Our analysis is neither device independent nor semi-device independent in the usual sense: no Hilbert-space dimension, quantum description, or model of the internal workings of the devices is assumed. Its central assumption is instead that Alice's and Bob's local tests are tomographically complete, so that the operational equivalences inferred from their statistics hold for the procedures at the laboratory boundary. We regard this as no less reasonable than the assumptions underlying device-independent protocols~\cite{spekkens2026against}. Device independence does not eliminate trust, it relocates it. Its security proof requires no-signaling to hold for every intervention available to the adversary, and what guarantees this is not the protocol but the assumption that the laboratories are shielded, so that no physical route other than the one relativity forbids connects Alice's setting to Bob's device. Tomographic completeness plays the same role here. The operational equivalences are assumed to hold for every procedure available to Eve, and what guarantees this is the assumption that the system leaving Alice's laboratory carries no operational degrees of freedom beyond those that have been characterized, which is a no-hidden-channel assumption on the designated carrier of exactly the kind that shielding imposes on the laboratory walls. In both cases the assumption is physical and defeasible, and fails through unmodelled structure in the honest parties' hardware. We therefore regard the two as different points in the same space of assumptions rather than as an ordered pair.

An important limitation of the present analysis is that the security bounds are derived against individual attacks. Extending the framework to general attacks requires controlling Eve's information across many uses of the protocol, where she may correlate her interventions between rounds, retain arbitrary post-quantum side information, and postpone a collective readout until after the public communication. In this regime, the single-round Csisz\'ar--K\"orner and reverse-Fano argument used above is not by itself a composable security proof. In quantum cryptography, such extensions are typically formulated in terms of smooth conditional min-entropy and privacy amplification, while device-independent proofs require additional structure to relate the observed single-round correlations to the multi-round entropy~\cite{Tomamichel2011fbe,masanes2011secure}. For genuinely post-quantum adversaries the issue is more subtle: privacy amplification can fail for sufficiently general no-signaling extensions~\cite{hanggi2013impossibility}, whereas composable security can be recovered under suitable additional causality constraints~\cite{masanes2009universally}. Determining the corresponding multi-round conditions under which the OPEQ-constrained adversarial model introduced here admits privacy amplification, and hence composable security against general attacks, is an important direction for future
work.

\begin{acknowledgments}
We thank Robert Spekkens, Matilde Baroni and Dominik Leichtle for fruitful discussions. More specially, we are grateful to Daniel Sun, Teresa Del Aguila Ferrandis and Rennan Gomes de Albuquerque for all the discussions and shared ideas during the winter school 2026 of the PSI Master program where this project started. We acknowledge the use of generative AI in the writing process of the manuscript as well as in the coding part of the project. Research at Perimeter Institute is supported in part by the Government of Canada through the Department of Innovation, Science and Economic Development and by the Province of Ontario through the Ministry of Colleges and Universities.
\end{acknowledgments}

\newpage

\bibliographystyle{apsrev4-2-wolfe}
\setlength{\bibsep}{.4\baselineskip plus .1\baselineskip minus .1\baselineskip}
\nocite{apsrev42Control}
\bibliography{references}

\begin{thebibliography}{54}%
\makeatletter
\providecommand \@ifxundefined [1]{%
 \@ifx{#1\undefined}
}%
\providecommand \@ifnum [1]{%
 \ifnum #1\expandafter \@firstoftwo
 \else \expandafter \@secondoftwo
 \fi
}%
\providecommand \@ifx [1]{%
 \ifx #1\expandafter \@firstoftwo
 \else \expandafter \@secondoftwo
 \fi
}%
\providecommand \natexlab [1]{#1}%
\providecommand \enquote  [1]{``#1''}%
\providecommand \bibnamefont  [1]{#1}%
\providecommand \bibfnamefont [1]{#1}%
\providecommand \citenamefont [1]{#1}%
\providecommand \href@noop [0]{\@secondoftwo}%
\providecommand \href [0]{\begingroup \@sanitize@url \@href}%
\providecommand \@href[1]{\@@startlink{#1}\@@href}%
\providecommand \@@href[1]{\endgroup#1\@@endlink}%
\providecommand \@sanitize@url [0]{\catcode `\\12\catcode `\$12\catcode `\&12\catcode `\#12\catcode `\^12\catcode `\_12\catcode `\%12\relax}%
\providecommand \@@startlink[1]{}%
\providecommand \@@endlink[0]{}%
\providecommand \url  [0]{\begingroup\@sanitize@url \@url }%
\providecommand \@url [1]{\endgroup\@href {#1}{\urlprefix }}%
\providecommand \urlprefix  [0]{URL }%
\providecommand \Eprint [0]{\href }%
\providecommand \doibase [0]{https://doi.org/}%
\providecommand \selectlanguage [0]{\@gobble}%
\providecommand \bibinfo  [0]{\@secondoftwo}%
\providecommand \bibfield  [0]{\@secondoftwo}%
\providecommand \translation [1]{[#1]}%
\providecommand \BibitemOpen [0]{}%
\providecommand \bibitemStop [0]{}%
\providecommand \bibitemNoStop [0]{.\EOS\space}%
\providecommand \EOS [0]{\spacefactor3000\relax}%
\providecommand \BibitemShut  [1]{\csname bibitem#1\endcsname}%
\let\auto@bib@innerbib\@empty
\bibitem [{\citenamefont {Gisin}\ \emph {et~al.}(2002)\citenamefont {Gisin}, \citenamefont {Ribordy}, \citenamefont {Tittel},\ and\ \citenamefont {Zbinden}}]{gisin2002quantum}%
  \BibitemOpen
  \bibfield  {author} {\bibinfo {author} {\bibfnamefont {N.}~\bibnamefont {Gisin}}, \bibinfo {author} {\bibfnamefont {G.}~\bibnamefont {Ribordy}}, \bibinfo {author} {\bibfnamefont {W.}~\bibnamefont {Tittel}},\ and\ \bibinfo {author} {\bibfnamefont {H.}~\bibnamefont {Zbinden}},\ }\bibfield  {title} {\enquote {\bibinfo {title} {{Quantum cryptography}},}\ }\href {https://doi.org/10.1103/RevModPhys.74.145} {\bibfield  {journal} {\bibinfo  {journal} {Rev. Mod. Phys.}\ }\textbf {\bibinfo {volume} {74}},\ \bibinfo {pages} {145} (\bibinfo {year} {2002})}\BibitemShut {NoStop}%
\bibitem [{\citenamefont {Scarani}\ \emph {et~al.}(2009)\citenamefont {Scarani}, \citenamefont {Bechmann-Pasquinucci}, \citenamefont {Cerf}, \citenamefont {Du{\v{s}}ek}, \citenamefont {L{\"u}tkenhaus},\ and\ \citenamefont {Peev}}]{scarani2009security}%
  \BibitemOpen
  \bibfield  {author} {\bibinfo {author} {\bibfnamefont {V.}~\bibnamefont {Scarani}}, \bibinfo {author} {\bibfnamefont {H.}~\bibnamefont {Bechmann-Pasquinucci}}, \bibinfo {author} {\bibfnamefont {N.~J.}\ \bibnamefont {Cerf}}, \bibinfo {author} {\bibfnamefont {M.}~\bibnamefont {Du{\v{s}}ek}}, \bibinfo {author} {\bibfnamefont {N.}~\bibnamefont {L{\"u}tkenhaus}},\ and\ \bibinfo {author} {\bibfnamefont {M.}~\bibnamefont {Peev}},\ }\bibfield  {title} {\enquote {\bibinfo {title} {{The security of practical quantum key distribution}},}\ }\href {https://doi.org/10.1103/RevModPhys.81.1301} {\bibfield  {journal} {\bibinfo  {journal} {Rev. Mod. Phys.}\ }\textbf {\bibinfo {volume} {81}},\ \bibinfo {pages} {1301} (\bibinfo {year} {2009})}\BibitemShut {NoStop}%
\bibitem [{\citenamefont {Acin}\ \emph {et~al.}(2006{\natexlab{a}})\citenamefont {Acin}, \citenamefont {Gisin},\ and\ \citenamefont {Masanes}}]{acin2006bell}%
  \BibitemOpen
  \bibfield  {author} {\bibinfo {author} {\bibfnamefont {A.}~\bibnamefont {Acin}}, \bibinfo {author} {\bibfnamefont {N.}~\bibnamefont {Gisin}},\ and\ \bibinfo {author} {\bibfnamefont {L.}~\bibnamefont {Masanes}},\ }\bibfield  {title} {\enquote {\bibinfo {title} {{From Bell’s theorem to secure quantum key distribution}},}\ }\href {https://doi.org/10.1103/PhysRevLett.97.120405} {\bibfield  {journal} {\bibinfo  {journal} {Phys. Rev. Lett.}\ }\textbf {\bibinfo {volume} {97}},\ \bibinfo {pages} {120405} (\bibinfo {year} {2006}{\natexlab{a}})}\BibitemShut {NoStop}%
\bibitem [{\citenamefont {Ekert}(1991)}]{ekert1991quantum}%
  \BibitemOpen
  \bibfield  {author} {\bibinfo {author} {\bibfnamefont {A.~K.}\ \bibnamefont {Ekert}},\ }\bibfield  {title} {\enquote {\bibinfo {title} {{Quantum cryptography based on Bell’s theorem}},}\ }\href {https://doi.org/10.1103/PhysRevLett.67.661} {\bibfield  {journal} {\bibinfo  {journal} {Phys. Rev. Lett.}\ }\textbf {\bibinfo {volume} {67}},\ \bibinfo {pages} {661} (\bibinfo {year} {1991})}\BibitemShut {NoStop}%
\bibitem [{\citenamefont {Masanes}\ \emph {et~al.}(2011)\citenamefont {Masanes}, \citenamefont {Pironio},\ and\ \citenamefont {Ac{\'\i}n}}]{masanes2011secure}%
  \BibitemOpen
  \bibfield  {author} {\bibinfo {author} {\bibfnamefont {L.}~\bibnamefont {Masanes}}, \bibinfo {author} {\bibfnamefont {S.}~\bibnamefont {Pironio}},\ and\ \bibinfo {author} {\bibfnamefont {A.}~\bibnamefont {Ac{\'\i}n}},\ }\bibfield  {title} {\enquote {\bibinfo {title} {{Secure device-independent quantum key distribution with causally independent measurement devices}},}\ }\href {https://www.nature.com/articles/ncomms1244} {\bibfield  {journal} {\bibinfo  {journal} {Nature Comm.}\ }\textbf {\bibinfo {volume} {2}},\ \bibinfo {pages} {238} (\bibinfo {year} {2011})}\BibitemShut {NoStop}%
\bibitem [{\citenamefont {Barrett}\ \emph {et~al.}(2005)\citenamefont {Barrett}, \citenamefont {Hardy},\ and\ \citenamefont {Kent}}]{barrett2005no}%
  \BibitemOpen
  \bibfield  {author} {\bibinfo {author} {\bibfnamefont {J.}~\bibnamefont {Barrett}}, \bibinfo {author} {\bibfnamefont {L.}~\bibnamefont {Hardy}},\ and\ \bibinfo {author} {\bibfnamefont {A.}~\bibnamefont {Kent}},\ }\bibfield  {title} {\enquote {\bibinfo {title} {{No signaling and quantum key distribution}},}\ }\href {https://doi.org/10.1103/PhysRevLett.95.010503} {\bibfield  {journal} {\bibinfo  {journal} {Phys. Rev. Lett.}\ }\textbf {\bibinfo {volume} {95}},\ \bibinfo {pages} {010503} (\bibinfo {year} {2005})}\BibitemShut {NoStop}%
\bibitem [{\citenamefont {Scarani}\ \emph {et~al.}(2006)\citenamefont {Scarani}, \citenamefont {Gisin}, \citenamefont {Brunner}, \citenamefont {Masanes}, \citenamefont {Pino},\ and\ \citenamefont {Ac{\'\i}n}}]{scarani2006secrecy}%
  \BibitemOpen
  \bibfield  {author} {\bibinfo {author} {\bibfnamefont {V.}~\bibnamefont {Scarani}}, \bibinfo {author} {\bibfnamefont {N.}~\bibnamefont {Gisin}}, \bibinfo {author} {\bibfnamefont {N.}~\bibnamefont {Brunner}}, \bibinfo {author} {\bibfnamefont {L.}~\bibnamefont {Masanes}}, \bibinfo {author} {\bibfnamefont {S.}~\bibnamefont {Pino}},\ and\ \bibinfo {author} {\bibfnamefont {A.}~\bibnamefont {Ac{\'\i}n}},\ }\bibfield  {title} {\enquote {\bibinfo {title} {{Secrecy extraction from no-signaling correlations}},}\ }\href {https://doi.org/10.1103/PhysRevA.74.042339} {\bibfield  {journal} {\bibinfo  {journal} {Phys. Rev. A—Atomic, Molecular, and Optical Physics}\ }\textbf {\bibinfo {volume} {74}},\ \bibinfo {pages} {042339} (\bibinfo {year} {2006})}\BibitemShut {NoStop}%
\bibitem [{\citenamefont {Acin}\ \emph {et~al.}(2006{\natexlab{b}})\citenamefont {Acin}, \citenamefont {Massar},\ and\ \citenamefont {Pironio}}]{acin2006efficient}%
  \BibitemOpen
  \bibfield  {author} {\bibinfo {author} {\bibfnamefont {A.}~\bibnamefont {Acin}}, \bibinfo {author} {\bibfnamefont {S.}~\bibnamefont {Massar}},\ and\ \bibinfo {author} {\bibfnamefont {S.}~\bibnamefont {Pironio}},\ }\bibfield  {title} {\enquote {\bibinfo {title} {{Efficient quantum key distribution secure against no-signalling eavesdroppers}},}\ }\href {https://iopscience.iop.org/article/10.1088/1367-2630/8/8/126} {\bibfield  {journal} {\bibinfo  {journal} {New J. Phys.}\ }\textbf {\bibinfo {volume} {8}},\ \bibinfo {pages} {126} (\bibinfo {year} {2006}{\natexlab{b}})}\BibitemShut {NoStop}%
\bibitem [{\citenamefont {Farkas}\ \emph {et~al.}(2021)\citenamefont {Farkas}, \citenamefont {Balanz{\'o}-Juand{\'o}}, \citenamefont {{\L}ukanowski}, \citenamefont {Ko{\l}ody{\'n}ski},\ and\ \citenamefont {Ac{\'\i}n}}]{farkas2021bell}%
  \BibitemOpen
  \bibfield  {author} {\bibinfo {author} {\bibfnamefont {M.}~\bibnamefont {Farkas}}, \bibinfo {author} {\bibfnamefont {M.}~\bibnamefont {Balanz{\'o}-Juand{\'o}}}, \bibinfo {author} {\bibfnamefont {K.}~\bibnamefont {{\L}ukanowski}}, \bibinfo {author} {\bibfnamefont {J.}~\bibnamefont {Ko{\l}ody{\'n}ski}},\ and\ \bibinfo {author} {\bibfnamefont {A.}~\bibnamefont {Ac{\'\i}n}},\ }\bibfield  {title} {\enquote {\bibinfo {title} {{Bell nonlocality is not sufficient for the security of standard device-independent quantum key distribution protocols}},}\ }\href {https://doi.org/10.1103/PhysRevLett.127.050503} {\bibfield  {journal} {\bibinfo  {journal} {Phys. Rev. Lett.}\ }\textbf {\bibinfo {volume} {127}},\ \bibinfo {pages} {050503} (\bibinfo {year} {2021})}\BibitemShut {NoStop}%
\bibitem [{\citenamefont {Liu}\ and\ \citenamefont {Ramanathan}(2025)}]{liu2025optimal}%
  \BibitemOpen
  \bibfield  {author} {\bibinfo {author} {\bibfnamefont {Y.}~\bibnamefont {Liu}}\ and\ \bibinfo {author} {\bibfnamefont {R.}~\bibnamefont {Ramanathan}},\ }\bibfield  {title} {\enquote {\bibinfo {title} {{Optimal and feasible contextuality-based randomness generation}},}\ }\href {https://doi.org/10.1103/cn3m-j743} {\bibfield  {journal} {\bibinfo  {journal} {Phys. Rev. Lett.}\ }\textbf {\bibinfo {volume} {135}},\ \bibinfo {pages} {170805} (\bibinfo {year} {2025})}\BibitemShut {NoStop}%
\bibitem [{\citenamefont {Singh}\ \emph {et~al.}(2017)\citenamefont {Singh}, \citenamefont {Bharti},\ and\ \citenamefont {Arvind}}]{singh2017quantum}%
  \BibitemOpen
  \bibfield  {author} {\bibinfo {author} {\bibfnamefont {J.}~\bibnamefont {Singh}}, \bibinfo {author} {\bibfnamefont {K.}~\bibnamefont {Bharti}},\ and\ \bibinfo {author} {\bibnamefont {Arvind}},\ }\bibfield  {title} {\enquote {\bibinfo {title} {{Quantum key distribution protocol based on contextuality monogamy}},}\ }\href {https://doi.org/10.1103/PhysRevA.95.062333} {\bibfield  {journal} {\bibinfo  {journal} {Phys. Rev. A}\ }\textbf {\bibinfo {volume} {95}},\ \bibinfo {pages} {062333} (\bibinfo {year} {2017})}\BibitemShut {NoStop}%
\bibitem [{\citenamefont {Meng}\ \emph {et~al.}(2025)\citenamefont {Meng}, \citenamefont {Saha}, \citenamefont {Mikkelsen}, \citenamefont {Henke}, \citenamefont {Wang}, \citenamefont {Bart}, \citenamefont {Ludwig}, \citenamefont {Lodahl}, \citenamefont {Cabello},\ and\ \citenamefont {Midolo}}]{meng2025contextuality}%
  \BibitemOpen
  \bibfield  {author} {\bibinfo {author} {\bibfnamefont {Y.}~\bibnamefont {Meng}}, \bibinfo {author} {\bibfnamefont {D.}~\bibnamefont {Saha}}, \bibinfo {author} {\bibfnamefont {M.~T.}\ \bibnamefont {Mikkelsen}}, \bibinfo {author} {\bibfnamefont {C.}~\bibnamefont {Henke}}, \bibinfo {author} {\bibfnamefont {Y.}~\bibnamefont {Wang}}, \bibinfo {author} {\bibfnamefont {N.}~\bibnamefont {Bart}}, \bibinfo {author} {\bibfnamefont {A.}~\bibnamefont {Ludwig}}, \bibinfo {author} {\bibfnamefont {P.}~\bibnamefont {Lodahl}}, \bibinfo {author} {\bibfnamefont {A.}~\bibnamefont {Cabello}},\ and\ \bibinfo {author} {\bibfnamefont {L.}~\bibnamefont {Midolo}},\ }\bibfield  {title} {\enquote {\bibinfo {title} {Contextuality-based quantum key distribution with deterministic single-photon sources},}\ }\href {https://arxiv.org/abs/2510.12761} {\bibfield  {journal} {\bibinfo  {journal} {arXiv preprint arXiv:2510.12761}\ } (\bibinfo {year} {2025})}\BibitemShut {NoStop}%
\bibitem [{\citenamefont {Cabello}\ \emph {et~al.}(2011)\citenamefont {Cabello}, \citenamefont {D’Ambrosio}, \citenamefont {Nagali},\ and\ \citenamefont {Sciarrino}}]{cabello2011hybrid}%
  \BibitemOpen
  \bibfield  {author} {\bibinfo {author} {\bibfnamefont {A.}~\bibnamefont {Cabello}}, \bibinfo {author} {\bibfnamefont {V.}~\bibnamefont {D’Ambrosio}}, \bibinfo {author} {\bibfnamefont {E.}~\bibnamefont {Nagali}},\ and\ \bibinfo {author} {\bibfnamefont {F.}~\bibnamefont {Sciarrino}},\ }\bibfield  {title} {\enquote {\bibinfo {title} {Hybrid ququart-encoded quantum cryptography protected by kochen-specker contextuality},}\ }\href {https://doi.org/10.1103/PhysRevA.84.030302} {\bibfield  {journal} {\bibinfo  {journal} {Physical Review A—Atomic, Molecular, and Optical Physics}\ }\textbf {\bibinfo {volume} {84}},\ \bibinfo {pages} {030302} (\bibinfo {year} {2011})}\BibitemShut {NoStop}%
\bibitem [{\citenamefont {Spekkens}(2005)}]{Spekkens2005}%
  \BibitemOpen
  \bibfield  {author} {\bibinfo {author} {\bibfnamefont {R.~W.}\ \bibnamefont {Spekkens}},\ }\bibfield  {title} {\enquote {\bibinfo {title} {{Contextuality for preparations, transformations, and unsharp measurements}},}\ }\href {https://doi.org/10.1103/PhysRevA.71.052108} {\bibfield  {journal} {\bibinfo  {journal} {Phys. Rev. A—Atomic, Molecular, and Optical Physics}\ }\textbf {\bibinfo {volume} {71}},\ \bibinfo {pages} {052108} (\bibinfo {year} {2005})}\BibitemShut {NoStop}%
\bibitem [{\citenamefont {Wright}\ and\ \citenamefont {Farkas}(2023)}]{wright2023invertible}%
  \BibitemOpen
  \bibfield  {author} {\bibinfo {author} {\bibfnamefont {V.~J.}\ \bibnamefont {Wright}}\ and\ \bibinfo {author} {\bibfnamefont {M.}~\bibnamefont {Farkas}},\ }\bibfield  {title} {\enquote {\bibinfo {title} {{Invertible map between Bell nonlocal and contextuality scenarios}},}\ }\href {https://doi.org/10.1103/PhysRevLett.131.220202} {\bibfield  {journal} {\bibinfo  {journal} {Phys. Rev. Lett.}\ }\textbf {\bibinfo {volume} {131}},\ \bibinfo {pages} {220202} (\bibinfo {year} {2023})}\BibitemShut {NoStop}%
\bibitem [{\citenamefont {Baroni}\ \emph {et~al.}(2026)\citenamefont {Baroni}, \citenamefont {Leichtle}, \citenamefont {{\v{S}}upi{\'c}}, \citenamefont {Markham},\ and\ \citenamefont {Quintino}}]{baroni2026composable}%
  \BibitemOpen
  \bibfield  {author} {\bibinfo {author} {\bibfnamefont {M.}~\bibnamefont {Baroni}}, \bibinfo {author} {\bibfnamefont {D.}~\bibnamefont {Leichtle}}, \bibinfo {author} {\bibfnamefont {I.}~\bibnamefont {{\v{S}}upi{\'c}}}, \bibinfo {author} {\bibfnamefont {D.}~\bibnamefont {Markham}},\ and\ \bibinfo {author} {\bibfnamefont {M.~T.}\ \bibnamefont {Quintino}},\ }\bibfield  {title} {\enquote {\bibinfo {title} {{Composable simultaneous purification: when all communication scenarios reduce to spatial correlations}},}\ }\href {https://arxiv.org/abs/2601.05158} {\bibfield  {journal} {\bibinfo  {journal} {arXiv preprint arXiv:2601.05158}\ } (\bibinfo {year} {2026})}\BibitemShut {NoStop}%
\bibitem [{\citenamefont {Tavakoli}\ and\ \citenamefont {Uola}(2020)}]{tavakoli2020measurement}%
  \BibitemOpen
  \bibfield  {author} {\bibinfo {author} {\bibfnamefont {A.}~\bibnamefont {Tavakoli}}\ and\ \bibinfo {author} {\bibfnamefont {R.}~\bibnamefont {Uola}},\ }\bibfield  {title} {\enquote {\bibinfo {title} {{Measurement incompatibility and steering are necessary and sufficient for operational contextuality}},}\ }\href {https://doi.org/10.1103/PhysRevResearch.2.013011} {\bibfield  {journal} {\bibinfo  {journal} {Physical Review Research}\ }\textbf {\bibinfo {volume} {2}},\ \bibinfo {pages} {013011} (\bibinfo {year} {2020})}\BibitemShut {NoStop}%
\bibitem [{\citenamefont {Yadavalli}\ and\ \citenamefont {Kunjwal}(2022)}]{yadavalli2022contextuality}%
  \BibitemOpen
  \bibfield  {author} {\bibinfo {author} {\bibfnamefont {S.~A.}\ \bibnamefont {Yadavalli}}\ and\ \bibinfo {author} {\bibfnamefont {R.}~\bibnamefont {Kunjwal}},\ }\bibfield  {title} {\enquote {\bibinfo {title} {{Contextuality in entanglement-assisted one-shot classical communication}},}\ }\href {https://quantum-journal.org/papers/q-2022-10-13-839/} {\bibfield  {journal} {\bibinfo  {journal} {Quantum}\ }\textbf {\bibinfo {volume} {6}},\ \bibinfo {pages} {839} (\bibinfo {year} {2022})}\BibitemShut {NoStop}%
\bibitem [{\citenamefont {Schmid}\ \emph {et~al.}(2024)\citenamefont {Schmid}, \citenamefont {Selby}, \citenamefont {Pusey},\ and\ \citenamefont {Spekkens}}]{Schmid2024structuretheorem}%
  \BibitemOpen
  \bibfield  {author} {\bibinfo {author} {\bibfnamefont {D.}~\bibnamefont {Schmid}}, \bibinfo {author} {\bibfnamefont {J.~H.}\ \bibnamefont {Selby}}, \bibinfo {author} {\bibfnamefont {M.~F.}\ \bibnamefont {Pusey}},\ and\ \bibinfo {author} {\bibfnamefont {R.~W.}\ \bibnamefont {Spekkens}},\ }\bibfield  {title} {\enquote {\bibinfo {title} {A structure theorem for generalized-noncontextual ontological models},}\ }\href {https://doi.org/10.22331/q-2024-03-14-1283} {\bibfield  {journal} {\bibinfo  {journal} {{Quantum}}\ }\textbf {\bibinfo {volume} {8}},\ \bibinfo {pages} {1283} (\bibinfo {year} {2024})}\BibitemShut {NoStop}%
\bibitem [{\citenamefont {Grabowecky}\ \emph {et~al.}(2022)\citenamefont {Grabowecky}, \citenamefont {Pollack}, \citenamefont {Cameron}, \citenamefont {Spekkens},\ and\ \citenamefont {Resch}}]{Grabowecky2022}%
  \BibitemOpen
  \bibfield  {author} {\bibinfo {author} {\bibfnamefont {M.~J.}\ \bibnamefont {Grabowecky}}, \bibinfo {author} {\bibfnamefont {C.~A.~J.}\ \bibnamefont {Pollack}}, \bibinfo {author} {\bibfnamefont {A.~R.}\ \bibnamefont {Cameron}}, \bibinfo {author} {\bibfnamefont {R.~W.}\ \bibnamefont {Spekkens}},\ and\ \bibinfo {author} {\bibfnamefont {K.~J.}\ \bibnamefont {Resch}},\ }\bibfield  {title} {\enquote {\bibinfo {title} {Experimentally bounding deviations from quantum theory for a photonic three-level system using theory-agnostic tomography},}\ }\href {https://doi.org/10.1103/PhysRevA.105.032204} {\bibfield  {journal} {\bibinfo  {journal} {Phys. Rev. A}\ }\textbf {\bibinfo {volume} {105}},\ \bibinfo {pages} {032204} (\bibinfo {year} {2022})}\BibitemShut {NoStop}%
\bibitem [{\citenamefont {Zhang}\ \emph {et~al.}(2026{\natexlab{a}})\citenamefont {Zhang}, \citenamefont {Schmid}, \citenamefont {Y{\=i}ng},\ and\ \citenamefont {Spekkens}}]{ZhangSchmidYingSpekkens2026}%
  \BibitemOpen
  \bibfield  {author} {\bibinfo {author} {\bibfnamefont {Y.}~\bibnamefont {Zhang}}, \bibinfo {author} {\bibfnamefont {D.}~\bibnamefont {Schmid}}, \bibinfo {author} {\bibfnamefont {Y.}~\bibnamefont {Y{\=i}ng}},\ and\ \bibinfo {author} {\bibfnamefont {R.~W.}\ \bibnamefont {Spekkens}},\ }\bibfield  {title} {\enquote {\bibinfo {title} {Reassessing the boundary between classical and nonclassical for individual quantum processes},}\ }\href {https://doi.org/10.1103/vqfz-wzjg} {\bibfield  {journal} {\bibinfo  {journal} {Phys. Rev. X}\ }\textbf {\bibinfo {volume} {16}},\ \bibinfo {pages} {021050} (\bibinfo {year} {2026}{\natexlab{a}})}\BibitemShut {NoStop}%
\bibitem [{\citenamefont {Schmid}\ and\ \citenamefont {Spekkens}(2018)}]{schmid2018contextual}%
  \BibitemOpen
  \bibfield  {author} {\bibinfo {author} {\bibfnamefont {D.}~\bibnamefont {Schmid}}\ and\ \bibinfo {author} {\bibfnamefont {R.~W.}\ \bibnamefont {Spekkens}},\ }\bibfield  {title} {\enquote {\bibinfo {title} {{Contextual advantage for state discrimination}},}\ }\href {https://doi.org/10.1103/PhysRevX.8.011015} {\bibfield  {journal} {\bibinfo  {journal} {Phys. Rev. X}\ }\textbf {\bibinfo {volume} {8}},\ \bibinfo {pages} {011015} (\bibinfo {year} {2018})}\BibitemShut {NoStop}%
\bibitem [{\citenamefont {Spekkens}\ \emph {et~al.}(2009)\citenamefont {Spekkens}, \citenamefont {Buzacott}, \citenamefont {Keehn}, \citenamefont {Toner},\ and\ \citenamefont {Pryde}}]{spekkens2009preparation}%
  \BibitemOpen
  \bibfield  {author} {\bibinfo {author} {\bibfnamefont {R.~W.}\ \bibnamefont {Spekkens}}, \bibinfo {author} {\bibfnamefont {D.~H.}\ \bibnamefont {Buzacott}}, \bibinfo {author} {\bibfnamefont {A.~J.}\ \bibnamefont {Keehn}}, \bibinfo {author} {\bibfnamefont {B.}~\bibnamefont {Toner}},\ and\ \bibinfo {author} {\bibfnamefont {G.~J.}\ \bibnamefont {Pryde}},\ }\bibfield  {title} {\enquote {\bibinfo {title} {{Preparation contextuality powers parity-oblivious multiplexing}},}\ }\href {https://doi.org/10.1103/PhysRevLett.102.010401} {\bibfield  {journal} {\bibinfo  {journal} {Phys. Rev. Lett.}\ }\textbf {\bibinfo {volume} {102}},\ \bibinfo {pages} {010401} (\bibinfo {year} {2009})}\BibitemShut {NoStop}%
\bibitem [{\citenamefont {Saha}\ and\ \citenamefont {Chaturvedi}(2019)}]{saha2019preparation}%
  \BibitemOpen
  \bibfield  {author} {\bibinfo {author} {\bibfnamefont {D.}~\bibnamefont {Saha}}\ and\ \bibinfo {author} {\bibfnamefont {A.}~\bibnamefont {Chaturvedi}},\ }\bibfield  {title} {\enquote {\bibinfo {title} {Preparation contextuality as an essential feature underlying quantum communication advantage},}\ }\href {https://doi.org/10.1103/PhysRevA.100.022108} {\bibfield  {journal} {\bibinfo  {journal} {Physical Review A}\ }\textbf {\bibinfo {volume} {100}},\ \bibinfo {pages} {022108} (\bibinfo {year} {2019})}\BibitemShut {NoStop}%
\bibitem [{\citenamefont {Troupe}\ and\ \citenamefont {Farinholt}(2015)}]{troupe2015contextuality}%
  \BibitemOpen
  \bibfield  {author} {\bibinfo {author} {\bibfnamefont {J.~E.}\ \bibnamefont {Troupe}}\ and\ \bibinfo {author} {\bibfnamefont {J.}~\bibnamefont {Farinholt}},\ }\bibfield  {title} {\enquote {\bibinfo {title} {{A contextuality based quantum key distribution protocol}},}\ }\href {https://arxiv.org/pdf/1512.02256} {\bibfield  {journal} {\bibinfo  {journal} {arXiv preprint arXiv:1512.02256}\ } (\bibinfo {year} {2015})}\BibitemShut {NoStop}%
\bibitem [{\citenamefont {Chaturvedi}\ \emph {et~al.}(2021)\citenamefont {Chaturvedi}, \citenamefont {Farkas},\ and\ \citenamefont {Wright}}]{chaturvedi2021characterising}%
  \BibitemOpen
  \bibfield  {author} {\bibinfo {author} {\bibfnamefont {A.}~\bibnamefont {Chaturvedi}}, \bibinfo {author} {\bibfnamefont {M.}~\bibnamefont {Farkas}},\ and\ \bibinfo {author} {\bibfnamefont {V.~J.}\ \bibnamefont {Wright}},\ }\bibfield  {title} {\enquote {\bibinfo {title} {{Characterising and bounding the set of quantum behaviours in contextuality scenarios}},}\ }\href {https://quantum-journal.org/papers/q-2021-06-29-484/} {\bibfield  {journal} {\bibinfo  {journal} {Quantum}\ }\textbf {\bibinfo {volume} {5}},\ \bibinfo {pages} {484} (\bibinfo {year} {2021})}\BibitemShut {NoStop}%
\bibitem [{\citenamefont {Mazurek}\ \emph {et~al.}(2021)\citenamefont {Mazurek}, \citenamefont {Pusey}, \citenamefont {Resch},\ and\ \citenamefont {Spekkens}}]{MazurekEtAl2021}%
  \BibitemOpen
  \bibfield  {author} {\bibinfo {author} {\bibfnamefont {M.~D.}\ \bibnamefont {Mazurek}}, \bibinfo {author} {\bibfnamefont {M.~F.}\ \bibnamefont {Pusey}}, \bibinfo {author} {\bibfnamefont {K.~J.}\ \bibnamefont {Resch}},\ and\ \bibinfo {author} {\bibfnamefont {R.~W.}\ \bibnamefont {Spekkens}},\ }\bibfield  {title} {\enquote {\bibinfo {title} {{Experimentally Bounding Deviations From Quantum Theory in the Landscape of Generalized Probabilistic Theories}},}\ }\href {https://doi.org/10.1103/PRXQuantum.2.020302} {\bibfield  {journal} {\bibinfo  {journal} {PRX Quantum}\ }\textbf {\bibinfo {volume} {2}},\ \bibinfo {pages} {020302} (\bibinfo {year} {2021})}\BibitemShut {NoStop}%
\bibitem [{\citenamefont {Zhang}\ \emph {et~al.}(2026{\natexlab{b}})\citenamefont {Zhang}, \citenamefont {Spodek}, \citenamefont {Schmid}, \citenamefont {Reid}, \citenamefont {Morrison}, \citenamefont {Jennewein}, \citenamefont {Resch},\ and\ \citenamefont {Spekkens}}]{zhang2026entanglement}%
  \BibitemOpen
  \bibfield  {author} {\bibinfo {author} {\bibfnamefont {Y.}~\bibnamefont {Zhang}}, \bibinfo {author} {\bibfnamefont {J.}~\bibnamefont {Spodek}}, \bibinfo {author} {\bibfnamefont {D.}~\bibnamefont {Schmid}}, \bibinfo {author} {\bibfnamefont {C.}~\bibnamefont {Reid}}, \bibinfo {author} {\bibfnamefont {L.~J.}\ \bibnamefont {Morrison}}, \bibinfo {author} {\bibfnamefont {T.}~\bibnamefont {Jennewein}}, \bibinfo {author} {\bibfnamefont {K.~J.}\ \bibnamefont {Resch}},\ and\ \bibinfo {author} {\bibfnamefont {R.~W.}\ \bibnamefont {Spekkens}},\ }\bibfield  {title} {\enquote {\bibinfo {title} {Entanglement certification using noncontextuality inequalities},}\ }\href {https://doi.org/10.1103/dxpr-wp6l} {\bibfield  {journal} {\bibinfo  {journal} {Physical Review X}\ }\textbf {\bibinfo {volume} {16}},\ \bibinfo {pages} {031057} (\bibinfo {year} {2026}{\natexlab{b}})}\BibitemShut {NoStop}%
\bibitem [{\citenamefont {Paw{\l}owski}\ and\ \citenamefont {Brunner}(2011)}]{PawlowskiBrunner2011}%
  \BibitemOpen
  \bibfield  {author} {\bibinfo {author} {\bibfnamefont {M.}~\bibnamefont {Paw{\l}owski}}\ and\ \bibinfo {author} {\bibfnamefont {N.}~\bibnamefont {Brunner}},\ }\bibfield  {title} {\enquote {\bibinfo {title} {{Semi-device-independent security of one-way quantum key distribution}},}\ }\href {https://doi.org/10.1103/PhysRevA.84.010302} {\bibfield  {journal} {\bibinfo  {journal} {Phys. Rev. A}\ }\textbf {\bibinfo {volume} {84}},\ \bibinfo {pages} {010302(R)} (\bibinfo {year} {2011})}\BibitemShut {NoStop}%
\bibitem [{\citenamefont {Woodhead}\ and\ \citenamefont {Pironio}(2015)}]{WoodheadPironio2015}%
  \BibitemOpen
  \bibfield  {author} {\bibinfo {author} {\bibfnamefont {E.}~\bibnamefont {Woodhead}}\ and\ \bibinfo {author} {\bibfnamefont {S.}~\bibnamefont {Pironio}},\ }\bibfield  {title} {\enquote {\bibinfo {title} {{Secrecy in Prepare-and-Measure Clauser-Horne-Shimony-Holt Tests with a Qubit Bound}},}\ }\href {https://doi.org/10.1103/PhysRevLett.115.150501} {\bibfield  {journal} {\bibinfo  {journal} {Phys. Rev. Lett.}\ }\textbf {\bibinfo {volume} {115}},\ \bibinfo {pages} {150501} (\bibinfo {year} {2015})}\BibitemShut {NoStop}%
\bibitem [{\citenamefont {Woodhead}(2016)}]{Woodhead2016}%
  \BibitemOpen
  \bibfield  {author} {\bibinfo {author} {\bibfnamefont {E.}~\bibnamefont {Woodhead}},\ }\bibfield  {title} {\enquote {\bibinfo {title} {{Semi device independence of the BB84 protocol}},}\ }\href {https://doi.org/10.1088/1367-2630/18/5/055010} {\bibfield  {journal} {\bibinfo  {journal} {New J. Phys.}\ }\textbf {\bibinfo {volume} {18}},\ \bibinfo {pages} {055010} (\bibinfo {year} {2016})}\BibitemShut {NoStop}%
\bibitem [{\citenamefont {Ioannou}\ \emph {et~al.}(2022{\natexlab{a}})\citenamefont {Ioannou}, \citenamefont {Pereira}, \citenamefont {Rusca}, \citenamefont {Gr{\"u}nenfelder}, \citenamefont {Boaron}, \citenamefont {Perrenoud}, \citenamefont {Abbott}, \citenamefont {Sekatski}, \citenamefont {Bancal}, \citenamefont {Maring}, \citenamefont {Zbinden},\ and\ \citenamefont {Brunner}}]{Ioannou2022RDI}%
  \BibitemOpen
  \bibfield  {author} {\bibinfo {author} {\bibfnamefont {M.}~\bibnamefont {Ioannou}}, \bibinfo {author} {\bibfnamefont {M.~A.}\ \bibnamefont {Pereira}}, \bibinfo {author} {\bibfnamefont {D.}~\bibnamefont {Rusca}}, \bibinfo {author} {\bibfnamefont {F.}~\bibnamefont {Gr{\"u}nenfelder}}, \bibinfo {author} {\bibfnamefont {A.}~\bibnamefont {Boaron}}, \bibinfo {author} {\bibfnamefont {M.}~\bibnamefont {Perrenoud}}, \bibinfo {author} {\bibfnamefont {A.~A.}\ \bibnamefont {Abbott}}, \bibinfo {author} {\bibfnamefont {P.}~\bibnamefont {Sekatski}}, \bibinfo {author} {\bibfnamefont {J.-D.}\ \bibnamefont {Bancal}}, \bibinfo {author} {\bibfnamefont {N.}~\bibnamefont {Maring}}, \bibinfo {author} {\bibfnamefont {H.}~\bibnamefont {Zbinden}},\ and\ \bibinfo {author} {\bibfnamefont {N.}~\bibnamefont {Brunner}},\ }\bibfield  {title} {\enquote {\bibinfo {title} {{Receiver-Device-Independent Quantum Key Distribution}},}\ }\href {https://doi.org/10.22331/q-2022-05-24-718} {\bibfield  {journal} {\bibinfo  {journal} {Quantum}\
  }\textbf {\bibinfo {volume} {6}},\ \bibinfo {pages} {718} (\bibinfo {year} {2022}{\natexlab{a}})}\BibitemShut {NoStop}%
\bibitem [{\citenamefont {Ioannou}\ \emph {et~al.}(2022{\natexlab{b}})\citenamefont {Ioannou}, \citenamefont {Sekatski}, \citenamefont {Abbott}, \citenamefont {Rosset}, \citenamefont {Bancal},\ and\ \citenamefont {Brunner}}]{Ioannou2022RDIProtocols}%
  \BibitemOpen
  \bibfield  {author} {\bibinfo {author} {\bibfnamefont {M.}~\bibnamefont {Ioannou}}, \bibinfo {author} {\bibfnamefont {P.}~\bibnamefont {Sekatski}}, \bibinfo {author} {\bibfnamefont {A.~A.}\ \bibnamefont {Abbott}}, \bibinfo {author} {\bibfnamefont {D.}~\bibnamefont {Rosset}}, \bibinfo {author} {\bibfnamefont {J.-D.}\ \bibnamefont {Bancal}},\ and\ \bibinfo {author} {\bibfnamefont {N.}~\bibnamefont {Brunner}},\ }\bibfield  {title} {\enquote {\bibinfo {title} {{Receiver-device-independent quantum key distribution protocols}},}\ }\href {https://doi.org/10.1088/1367-2630/ac71bc} {\bibfield  {journal} {\bibinfo  {journal} {New J. Phys.}\ }\textbf {\bibinfo {volume} {24}},\ \bibinfo {pages} {063006} (\bibinfo {year} {2022}{\natexlab{b}})}\BibitemShut {NoStop}%
\bibitem [{\citenamefont {Van~Himbeeck}\ \emph {et~al.}(2017)\citenamefont {Van~Himbeeck}, \citenamefont {Woodhead}, \citenamefont {Cerf}, \citenamefont {Garc{\'{i}}a-Patr{\'{o}}n},\ and\ \citenamefont {Pironio}}]{VanHimbeeck2017}%
  \BibitemOpen
  \bibfield  {author} {\bibinfo {author} {\bibfnamefont {T.}~\bibnamefont {Van~Himbeeck}}, \bibinfo {author} {\bibfnamefont {E.}~\bibnamefont {Woodhead}}, \bibinfo {author} {\bibfnamefont {N.~J.}\ \bibnamefont {Cerf}}, \bibinfo {author} {\bibfnamefont {R.}~\bibnamefont {Garc{\'{i}}a-Patr{\'{o}}n}},\ and\ \bibinfo {author} {\bibfnamefont {S.}~\bibnamefont {Pironio}},\ }\bibfield  {title} {\enquote {\bibinfo {title} {{Semi-device-independent framework based on natural physical assumptions}},}\ }\href {https://doi.org/10.22331/q-2017-11-18-33} {\bibfield  {journal} {\bibinfo  {journal} {{Quantum}}\ }\textbf {\bibinfo {volume} {1}},\ \bibinfo {pages} {33} (\bibinfo {year} {2017})}\BibitemShut {NoStop}%
\bibitem [{\citenamefont {Ac{\'\i}n}\ \emph {et~al.}(2007)\citenamefont {Ac{\'\i}n}, \citenamefont {Brunner}, \citenamefont {Gisin}, \citenamefont {Massar}, \citenamefont {Pironio},\ and\ \citenamefont {Scarani}}]{acin2007device}%
  \BibitemOpen
  \bibfield  {author} {\bibinfo {author} {\bibfnamefont {A.}~\bibnamefont {Ac{\'\i}n}}, \bibinfo {author} {\bibfnamefont {N.}~\bibnamefont {Brunner}}, \bibinfo {author} {\bibfnamefont {N.}~\bibnamefont {Gisin}}, \bibinfo {author} {\bibfnamefont {S.}~\bibnamefont {Massar}}, \bibinfo {author} {\bibfnamefont {S.}~\bibnamefont {Pironio}},\ and\ \bibinfo {author} {\bibfnamefont {V.}~\bibnamefont {Scarani}},\ }\bibfield  {title} {\enquote {\bibinfo {title} {{Device-independent security of quantum cryptography against collective attacks}},}\ }\href {https://doi.org/10.1103/PhysRevLett.98.230501} {\bibfield  {journal} {\bibinfo  {journal} {Phys. Rev. Lett.}\ }\textbf {\bibinfo {volume} {98}},\ \bibinfo {pages} {230501} (\bibinfo {year} {2007})}\BibitemShut {NoStop}%
\bibitem [{\citenamefont {Primaatmaja}\ \emph {et~al.}(2023)\citenamefont {Primaatmaja}, \citenamefont {Goh}, \citenamefont {Tan}, \citenamefont {Khoo}, \citenamefont {Ghorai},\ and\ \citenamefont {Lim}}]{Primaatmaja2023security}%
  \BibitemOpen
  \bibfield  {author} {\bibinfo {author} {\bibfnamefont {I.~W.}\ \bibnamefont {Primaatmaja}}, \bibinfo {author} {\bibfnamefont {K.~T.}\ \bibnamefont {Goh}}, \bibinfo {author} {\bibfnamefont {E.~Y.-Z.}\ \bibnamefont {Tan}}, \bibinfo {author} {\bibfnamefont {J.~T.-F.}\ \bibnamefont {Khoo}}, \bibinfo {author} {\bibfnamefont {S.}~\bibnamefont {Ghorai}},\ and\ \bibinfo {author} {\bibfnamefont {C.~C.-W.}\ \bibnamefont {Lim}},\ }\bibfield  {title} {\enquote {\bibinfo {title} {{Security of device-independent quantum key distribution protocols: a review}},}\ }\href {https://doi.org/10.22331/q-2023-03-02-932} {\bibfield  {journal} {\bibinfo  {journal} {{Quantum}}\ }\textbf {\bibinfo {volume} {7}},\ \bibinfo {pages} {932} (\bibinfo {year} {2023})}\BibitemShut {NoStop}%
\bibitem [{\citenamefont {Csisz{\'a}r}\ and\ \citenamefont {Korner}(1978)}]{csiszar1978broadcast}%
  \BibitemOpen
  \bibfield  {author} {\bibinfo {author} {\bibfnamefont {I.}~\bibnamefont {Csisz{\'a}r}}\ and\ \bibinfo {author} {\bibfnamefont {J.}~\bibnamefont {Korner}},\ }\bibfield  {title} {\enquote {\bibinfo {title} {{Broadcast channels with confidential messages}},}\ }\href {https://ieeexplore.ieee.org/document/1055892} {\bibfield  {journal} {\bibinfo  {journal} {IEEE transactions on information theory}\ }\textbf {\bibinfo {volume} {24}},\ \bibinfo {pages} {339} (\bibinfo {year} {1978})}\BibitemShut {NoStop}%
\bibitem [{\citenamefont {Feder}\ and\ \citenamefont {Merhav}(1994)}]{feder1994relations}%
  \BibitemOpen
  \bibfield  {author} {\bibinfo {author} {\bibfnamefont {M.}~\bibnamefont {Feder}}\ and\ \bibinfo {author} {\bibfnamefont {N.}~\bibnamefont {Merhav}},\ }\bibfield  {title} {\enquote {\bibinfo {title} {{Relations between entropy and error probability}},}\ }\href {https://ieeexplore.ieee.org/document/272494} {\bibfield  {journal} {\bibinfo  {journal} {IEEE Transactions on Information theory}\ }\textbf {\bibinfo {volume} {40}},\ \bibinfo {pages} {259} (\bibinfo {year} {1994})}\BibitemShut {NoStop}%
\bibitem [{\citenamefont {Rioul}(2023)}]{rioul2023interplay}%
  \BibitemOpen
  \bibfield  {author} {\bibinfo {author} {\bibfnamefont {O.}~\bibnamefont {Rioul}},\ }\bibfield  {title} {\enquote {\bibinfo {title} {{The interplay between error, total variation, alpha-entropy and guessing: Fano and Pinsker direct and reverse inequalities}},}\ }\href {https://www.mdpi.com/1099-4300/25/7/978} {\bibfield  {journal} {\bibinfo  {journal} {Entropy}\ }\textbf {\bibinfo {volume} {25}},\ \bibinfo {pages} {978} (\bibinfo {year} {2023})}\BibitemShut {NoStop}%
\bibitem [{\citenamefont {Tomamichel}\ \emph {et~al.}(2011)\citenamefont {Tomamichel}, \citenamefont {Schaffner}, \citenamefont {Smith},\ and\ \citenamefont {Renner}}]{Tomamichel2011fbe}%
  \BibitemOpen
  \bibfield  {author} {\bibinfo {author} {\bibfnamefont {M.}~\bibnamefont {Tomamichel}}, \bibinfo {author} {\bibfnamefont {C.}~\bibnamefont {Schaffner}}, \bibinfo {author} {\bibfnamefont {A.}~\bibnamefont {Smith}},\ and\ \bibinfo {author} {\bibfnamefont {R.}~\bibnamefont {Renner}},\ }\bibfield  {title} {\enquote {\bibinfo {title} {{Leftover Hashing Against Quantum Side Information}},}\ }\href {https://doi.org/10.1109/TIT.2011.2158473} {\bibfield  {journal} {\bibinfo  {journal} {IEEE Trans. Info. Theor.}\ }\textbf {\bibinfo {volume} {57}},\ \bibinfo {pages} {5524} (\bibinfo {year} {2011})}\BibitemShut {NoStop}%
\bibitem [{\citenamefont {Tomamichel}\ \emph {et~al.}(2012)\citenamefont {Tomamichel}, \citenamefont {Lim}, \citenamefont {Gisin},\ and\ \citenamefont {Renner}}]{tomamichel2012tight}%
  \BibitemOpen
  \bibfield  {author} {\bibinfo {author} {\bibfnamefont {M.}~\bibnamefont {Tomamichel}}, \bibinfo {author} {\bibfnamefont {C.~C.~W.}\ \bibnamefont {Lim}}, \bibinfo {author} {\bibfnamefont {N.}~\bibnamefont {Gisin}},\ and\ \bibinfo {author} {\bibfnamefont {R.}~\bibnamefont {Renner}},\ }\bibfield  {title} {\enquote {\bibinfo {title} {{Tight finite-key analysis for quantum cryptography}},}\ }\href {https://www.nature.com/articles/ncomms1631} {\bibfield  {journal} {\bibinfo  {journal} {Nature Comm.}\ }\textbf {\bibinfo {volume} {3}},\ \bibinfo {pages} {634} (\bibinfo {year} {2012})}\BibitemShut {NoStop}%
\bibitem [{\citenamefont {Wolfe}\ and\ \citenamefont {collaborators}(2026)}]{ContextualityQKDGitHub}%
  \BibitemOpen
  \bibfield  {author} {\bibinfo {author} {\bibfnamefont {E.}~\bibnamefont {Wolfe}}\ and\ \bibinfo {author} {\bibnamefont {collaborators}},\ }\href@noop {} {\enquote {\bibinfo {title} {{contextualityqkd: Bob-Outcome QKD from Contextuality}},}\ }\bibinfo {howpublished} {\url{https://github.com/eliewolfe/RandomnessWinterSchool2026}} (\bibinfo {year} {2026}),\ \bibinfo {note} {gitHub repository}\BibitemShut {NoStop}%
\bibitem [{\citenamefont {Mazurek}\ \emph {et~al.}(2016)\citenamefont {Mazurek}, \citenamefont {Pusey}, \citenamefont {Kunjwal}, \citenamefont {Resch},\ and\ \citenamefont {Spekkens}}]{mazurek2016experimental}%
  \BibitemOpen
  \bibfield  {author} {\bibinfo {author} {\bibfnamefont {M.~D.}\ \bibnamefont {Mazurek}}, \bibinfo {author} {\bibfnamefont {M.~F.}\ \bibnamefont {Pusey}}, \bibinfo {author} {\bibfnamefont {R.}~\bibnamefont {Kunjwal}}, \bibinfo {author} {\bibfnamefont {K.~J.}\ \bibnamefont {Resch}},\ and\ \bibinfo {author} {\bibfnamefont {R.~W.}\ \bibnamefont {Spekkens}},\ }\bibfield  {title} {\enquote {\bibinfo {title} {An experimental test of noncontextuality without unphysical idealizations},}\ }\href {https://www.nature.com/articles/ncomms11780} {\bibfield  {journal} {\bibinfo  {journal} {Nature communications}\ }\textbf {\bibinfo {volume} {7}},\ \bibinfo {pages} {ncomms11780} (\bibinfo {year} {2016})}\BibitemShut {NoStop}%
\bibitem [{\citenamefont {Cabello}\ \emph {et~al.}(1996)\citenamefont {Cabello}, \citenamefont {Estebaranz},\ and\ \citenamefont {Garc{\'\i}a-Alcaine}}]{cabello1996bell}%
  \BibitemOpen
  \bibfield  {author} {\bibinfo {author} {\bibfnamefont {A.}~\bibnamefont {Cabello}}, \bibinfo {author} {\bibfnamefont {J.}~\bibnamefont {Estebaranz}},\ and\ \bibinfo {author} {\bibfnamefont {G.}~\bibnamefont {Garc{\'\i}a-Alcaine}},\ }\bibfield  {title} {\enquote {\bibinfo {title} {Bell-kochen-specker theorem: A proof with 18 vectors},}\ }\href {https://www.sciencedirect.com/science/article/pii/037596019600134X?via%3Dihub} {\bibfield  {journal} {\bibinfo  {journal} {Physics Letters A}\ }\textbf {\bibinfo {volume} {212}},\ \bibinfo {pages} {183} (\bibinfo {year} {1996})}\BibitemShut {NoStop}%
\bibitem [{\citenamefont {Peres}(1991)}]{peres1991two}%
  \BibitemOpen
  \bibfield  {author} {\bibinfo {author} {\bibfnamefont {A.}~\bibnamefont {Peres}},\ }\bibfield  {title} {\enquote {\bibinfo {title} {Two simple proofs of the kochen-specker theorem},}\ }\href {https://iopscience.iop.org/article/10.1088/0305-4470/24/4/003} {\bibfield  {journal} {\bibinfo  {journal} {Journal of Physics A: Mathematical and General}\ }\textbf {\bibinfo {volume} {24}},\ \bibinfo {pages} {L175} (\bibinfo {year} {1991})}\BibitemShut {NoStop}%
\bibitem [{\citenamefont {Pusey}(2018)}]{Pusey2018}%
  \BibitemOpen
  \bibfield  {author} {\bibinfo {author} {\bibfnamefont {M.~F.}\ \bibnamefont {Pusey}},\ }\bibfield  {title} {\enquote {\bibinfo {title} {{Robust preparation noncontextuality inequalities in the simplest scenario}},}\ }\href {https://doi.org/10.1103/PhysRevA.98.022112} {\bibfield  {journal} {\bibinfo  {journal} {Phys. Rev. A}\ }\textbf {\bibinfo {volume} {98}},\ \bibinfo {pages} {022112} (\bibinfo {year} {2018})}\BibitemShut {NoStop}%
\bibitem [{\citenamefont {Lo}\ \emph {et~al.}(2005)\citenamefont {Lo}, \citenamefont {Chau},\ and\ \citenamefont {Ardehali}}]{LoChauArdehali2005}%
  \BibitemOpen
  \bibfield  {author} {\bibinfo {author} {\bibfnamefont {H.-K.}\ \bibnamefont {Lo}}, \bibinfo {author} {\bibfnamefont {H.~F.}\ \bibnamefont {Chau}},\ and\ \bibinfo {author} {\bibfnamefont {M.}~\bibnamefont {Ardehali}},\ }\bibfield  {title} {\enquote {\bibinfo {title} {{Efficient Quantum Key Distribution Scheme and a Proof of Its Unconditional Security}},}\ }\href {https://doi.org/10.1007/s00145-004-0142-y} {\bibfield  {journal} {\bibinfo  {journal} {J. Cryptology}\ }\textbf {\bibinfo {volume} {18}},\ \bibinfo {pages} {133} (\bibinfo {year} {2005})}\BibitemShut {NoStop}%
\bibitem [{\citenamefont {Xu}\ \emph {et~al.}(2020)\citenamefont {Xu}, \citenamefont {Ma}, \citenamefont {Zhang}, \citenamefont {Lo},\ and\ \citenamefont {Pan}}]{Xu2020}%
  \BibitemOpen
  \bibfield  {author} {\bibinfo {author} {\bibfnamefont {F.}~\bibnamefont {Xu}}, \bibinfo {author} {\bibfnamefont {X.}~\bibnamefont {Ma}}, \bibinfo {author} {\bibfnamefont {Q.}~\bibnamefont {Zhang}}, \bibinfo {author} {\bibfnamefont {H.-K.}\ \bibnamefont {Lo}},\ and\ \bibinfo {author} {\bibfnamefont {J.-W.}\ \bibnamefont {Pan}},\ }\bibfield  {title} {\enquote {\bibinfo {title} {{Secure quantum key distribution with realistic devices}},}\ }\href {https://doi.org/10.1103/RevModPhys.92.025002} {\bibfield  {journal} {\bibinfo  {journal} {Rev. Mod. Phys.}\ }\textbf {\bibinfo {volume} {92}},\ \bibinfo {pages} {025002} (\bibinfo {year} {2020})}\BibitemShut {NoStop}%
\bibitem [{\citenamefont {Spekkens}(2026)}]{spekkens2026against}%
  \BibitemOpen
  \bibfield  {author} {\bibinfo {author} {\bibfnamefont {R.~W.}\ \bibnamefont {Spekkens}},\ }\bibfield  {title} {\enquote {\bibinfo {title} {Against device independence},}\ } (\bibinfo {year} {2026}),\ \bibinfo {note} {forthcoming}\BibitemShut {NoStop}%
\bibitem [{\citenamefont {H{\"a}nggi}\ \emph {et~al.}(2013)\citenamefont {H{\"a}nggi}, \citenamefont {Renner},\ and\ \citenamefont {Wolf}}]{hanggi2013impossibility}%
  \BibitemOpen
  \bibfield  {author} {\bibinfo {author} {\bibfnamefont {E.}~\bibnamefont {H{\"a}nggi}}, \bibinfo {author} {\bibfnamefont {R.}~\bibnamefont {Renner}},\ and\ \bibinfo {author} {\bibfnamefont {S.}~\bibnamefont {Wolf}},\ }\bibfield  {title} {\enquote {\bibinfo {title} {{The impossibility of non-signaling privacy amplification}},}\ }\href {https://www.sciencedirect.com/science/article/pii/S0304397512011073} {\bibfield  {journal} {\bibinfo  {journal} {Theor. Comp. Sci.}\ }\textbf {\bibinfo {volume} {486}},\ \bibinfo {pages} {27} (\bibinfo {year} {2013})}\BibitemShut {NoStop}%
\bibitem [{\citenamefont {Masanes}(2009)}]{masanes2009universally}%
  \BibitemOpen
  \bibfield  {author} {\bibinfo {author} {\bibfnamefont {L.}~\bibnamefont {Masanes}},\ }\bibfield  {title} {\enquote {\bibinfo {title} {{Universally composable privacy amplification from causality constraints}},}\ }\href {https://doi.org/10.1103/PhysRevLett.102.140501} {\bibfield  {journal} {\bibinfo  {journal} {Phys. Rev. Lett.}\ }\textbf {\bibinfo {volume} {102}},\ \bibinfo {pages} {140501} (\bibinfo {year} {2009})}\BibitemShut {NoStop}%
\bibitem [{\citenamefont {Barrett}(2007)}]{barrett2007information}%
  \BibitemOpen
  \bibfield  {author} {\bibinfo {author} {\bibfnamefont {J.}~\bibnamefont {Barrett}},\ }\bibfield  {title} {\enquote {\bibinfo {title} {Information processing in generalized probabilistic theories},}\ }\href {https://doi.org/10.1103/PhysRevA.75.032304} {\bibfield  {journal} {\bibinfo  {journal} {Physical Review A—Atomic, Molecular, and Optical Physics}\ }\textbf {\bibinfo {volume} {75}},\ \bibinfo {pages} {032304} (\bibinfo {year} {2007})}\BibitemShut {NoStop}%
\bibitem [{\citenamefont {Schmid}\ \emph {et~al.}(2018)\citenamefont {Schmid}, \citenamefont {Spekkens},\ and\ \citenamefont {Wolfe}}]{schmid2018all}%
  \BibitemOpen
  \bibfield  {author} {\bibinfo {author} {\bibfnamefont {D.}~\bibnamefont {Schmid}}, \bibinfo {author} {\bibfnamefont {R.~W.}\ \bibnamefont {Spekkens}},\ and\ \bibinfo {author} {\bibfnamefont {E.}~\bibnamefont {Wolfe}},\ }\bibfield  {title} {\enquote {\bibinfo {title} {All the noncontextuality inequalities for arbitrary prepare-and-measure experiments with respect to any fixed set of operational equivalences},}\ }\href {https://doi.org/10.1103/PhysRevA.97.062103} {\bibfield  {journal} {\bibinfo  {journal} {Physical Review A}\ }\textbf {\bibinfo {volume} {97}},\ \bibinfo {pages} {062103} (\bibinfo {year} {2018})}\BibitemShut {NoStop}%
\bibitem [{\citenamefont {Maurer}(1993)}]{maurer1993secret}%
  \BibitemOpen
  \bibfield  {author} {\bibinfo {author} {\bibfnamefont {U.~M.}\ \bibnamefont {Maurer}},\ }\bibfield  {title} {\enquote {\bibinfo {title} {Secret key agreement by public discussion from common information},}\ }\href {https://ieeexplore.ieee.org/document/256484} {\bibfield  {journal} {\bibinfo  {journal} {IEEE transactions on information theory}\ }\textbf {\bibinfo {volume} {39}},\ \bibinfo {pages} {733} (\bibinfo {year} {1993})}\BibitemShut {NoStop}%
\end{thebibliography}%

\appendix
\section{Decomposition of the adversarial optimization into fixed-readout linear programs}
\label{app:fixed_setting}

In this appendix, we prove that the optimization in Eq.~\eqref{eq:adversarial_LP} separates into one linear program per key-generating setting. We also show that the resulting optimal strategy is not an artifact of leaving the underlying theory unspecified, it is realizable within an explicit generalized probabilistic theory (GPT).

\subsection{Reduction to fixed-readout programs}

Recall that in an individual attack, Eve interacts with the system leaving Alice's laboratory before Bob's setting is announced. Thus, this interaction is independent of $y$ and produces a system forwarded to Bob together with a system retained by Eve. After Bob has measured and announced $y$, Eve chooses a readout $y'$ on her retained system and obtains $e$. The complete single-round behaviour is $p(b,e|x,y,y')$, where $y'$ labels only Eve's delayed readout. For each readout choice $y'$ we call
\begin{equation}
    P^{y'}(b,e|x,y) \;\equiv\; p(b,e|x,y,y')
    \label{eq:app_slice}
\end{equation}
the \emph{$y'$-slice} of the behaviour. Importantly, the constraints of the adversarial model act slice by slice, that is, there are no constraints between different slices. Let $\mathcal F$ denote the set of behaviours $P(b,e|x,y)$ satisfying
\begin{subequations}
\label{eq:app_F}
\begin{align}
    &P(b,e|x,y)\ge 0, \qquad \textstyle\sum_{b,e}P(b,e|x,y)=1, \\
    &\textstyle\sum_{e} P(b,e|x,y) = p_{\mathrm{obs}}(b|x,y), \label{eq:app_F_marg}\\
    &\textstyle\sum_{x}\alpha_x P(b,e|x,y) = \sum_{x}\alpha'_x P(b,e|x,y)\label{eq:app_F_prep}\\
    \quad&\forall\text{ preparation OPEQ}\nonumber,\\
    &\textstyle\sum_{b,y}\beta_{by} P(b,e|x,y) = \sum_{b,y}\beta'_{by} P(b,e|x,y)\label{eq:app_F_meas}\\
    \quad&\forall\text{ measurement OPEQ}\nonumber,\\
    &\textstyle\sum_{b} P(b,e|x,y) = \sum_{b} P(b,e|x,\tilde y)
    \qquad \forall\, y,\tilde y . \label{eq:app_F_NS}
\end{align}
\end{subequations}
Eq.~\eqref{eq:app_F_NS} expresses the no-signaling relation from $y$ to Eve, that is, for a fixed setting for Eve, Bob's setting cannot influence Eve's outcome. Note that it is implied by the trivial measurement OPEQs, i.e. normalization, $\sum_b M^b_y\simeq\sum_b M^b_{\tilde y}$, but we keep it explicit. A delayed-readout behaviour is admissible iff every slice lies in $\mathcal F$,
\begin{equation}
    P^{y'}\in\mathcal F \qquad\forall\,y' .
    \label{eq:app_admissible}
\end{equation}
Note that the delayed-readout model requires, in addition, that Eve's readout choice cannot influence Bob's statistics, $\sum_e p(b,e|x,y,y')=\sum_e p(b,e|x,y,\tilde y')$ for all $y',\tilde y'$. This no-signaling condition from $y'$ to $b$ need not be imposed separately as Eq.~\eqref{eq:app_F_marg} fixes the Bob marginal of every slice to the same observed distribution $p_{\mathrm{obs}}(b|x,y)$, so it holds automatically once Eq.~\eqref{eq:app_admissible} does.

Let us now describe the objective of the fixed-setting linear program. For a target setting $y^\star$, define the linear functional
\begin{equation}
    f_{y^\star}[P] \;=\; \sum_{x\in\mathcal X^{y^\star}_{\mathrm{key}}}
    q(x|y^\star)\; P\big(e=\kappa(x,y^\star)\,\big|\,x,y^\star\big).
    \label{eq:app_f}
\end{equation}
Then, the \emph{fixed-setting linear program} for $y^\star$ is
\begin{equation}
    P^{(y^\star)}_{\mathrm{guess}} \;=\; \max_{P\in\mathcal F} f_{y^\star}[P],
    \label{eq:app_fixed_LP}
\end{equation}
and we write $P^{y^\star}_\star\in\mathcal F$ for the full solution of the previous linear program. In other words, $P^{(y^\star)}_{\mathrm{guess}}$ is a number, the optimal value of the program with target setting $y^\star$, whereas $P^{y^\star}_\star$ is a full probability table $P^{y^\star}_\star(b,e|x,y)$, defined for every setting $y$ (not only for $y=y^\star$), at which that value is attained. The entries with $y\neq y^\star$ do not enter the objective but are needed for the OPEQs and the no-signaling constraint, which sum over all settings. Since Eve chooses her readout after Bob announces $y$, a round with setting $y$ contributes through the slice $P^{y'}$ with $y'=y$. Thus, her overall guessing probability is %
\begin{equation}
    P_{\mathrm{guess}}(K|E,Y) \;=\; \max_{\{P^{y'}\}}\;\sum_{y\in\mathcal Y_{\mathrm{key}}} q(y)\, f_y[P^{y}],
    \label{eq:app_full_LP}
\end{equation}
subject to Eq.~\eqref{eq:app_admissible}. 

The previous observations allow us to state the following proposition.
\begin{proposition}
\label{prop:app_separation}
Eve's optimal delayed-readout attack is obtained by optimizing her readout independently for each key-generating setting. That is, her overall guessing probability is the $q(y)$-weighted average of the fixed-readout optimal values,
\begin{equation}
    P_{\mathrm{guess}}(K|E,Y)=\sum_{y\in\mathcal Y_{\mathrm{key}}} q(y)\,P^{(y)}_{\mathrm{guess}},
    \label{eq:app_separation}
\end{equation}
and this value is attained by the strategy in which, for each $y\in\mathcal Y_{\mathrm{key}}$, Eve's $y$-slice is the optimizer $P^{y}_\star$ of the fixed-readout program for that setting.
\end{proposition}

Intuitively, since Eve chooses her readout only after $y$ is announced, and the adversarial model imposes no constraint linking her alternative readouts, she can treat each announced setting as a separate problem. The proof uses exactly these two facts: the admissible set in Eq.~\eqref{eq:app_admissible} is a product of identical constraint sets, one per slice, and the objective is a sum of terms each depending on a single slice. Consequently, $P_{\mathrm{guess}}(K|E,Y)$ is obtained by solving $|\mathcal Y_{\mathrm{key}}|$ independent linear programs, each over a single slice, rather than one joint program over all slices.

The separation is a direct consequence of the absence of inter-slice constraints. In a theory that relates Eve's alternative readouts of the same retained system (for instance, quantum theory) the slices could not be chosen independently, and the joint program would in general yield a strictly smaller value; see Sec.~\ref{app:gpt_realization}.

\subsection{A GPT realizing the optimal strategy}
\label{app:gpt_realization}

The separation relies on the absence of constraints between the different slices, and one may ask whether this makes the adversary unphysically strong. In any concrete physical theory, the slices $\{P^{y'}\}$ are not independent objects, they are the statistics of different measurements performed on the same system that Eve retained. So, it may be the case that the state and effect spaces of that theory do not allow every collection of slices, each individually in $\mathcal F$, to arise from a single state in this way. We now show that no such restriction can arise from the requirement of belonging to a GPT. The optimal family $\{P^{y}_\star\}$ is realizable by a single state and a single set of measurements in an explicit GPT built along the lines of Ref.~\cite{barrett2007information}. Consequently, the value in Proposition~\ref{prop:app_separation} is attained by an adversary in a definite theory, and the guessing probability is an equality rather than merely an upper bound.

We recall the general framework of Ref.~\cite{barrett2007information}. A GPT system is specified by a convex state space $\Omega$ in a real vector space. Measurements are sets of linear functionals (effects) that are nonnegative on $\Omega$ and sum to the unit effect. Transformations are linear maps between the vector spaces of two systems that send the state space of the input system into that of the output system (a given theory may allow all such maps or only a subset of them). Without loss of generality, when a system is described through a set of fiducial measurements, its state is identified with the table of outcome probabilities of those measurements, so two preparations with the same table are the same state. For composite systems the framework admits several rules. The most permissive is the maximal tensor product, in which the joint state space $\Omega_1\otimes_{\max}\Omega_2$ consists of all no-signaling tables of joint fiducial probabilities whose conditional states, obtained by conditioning on the outcome of a fiducial measurement on one system, lie in the state space of the other.

These are the only ingredients we use. The theory constructed below makes three specific choices within this framework: the state spaces of Bob's and Eve's systems, the composition rule (we take the maximal tensor product), and the set of allowed transformations, which we take to include the map $\mathcal T$ defined below. 

Fix an optimal family $\{P^{y'}_\star\}_{y'\in\mathcal Y}$ and define, for each $x$, the joint table
\begin{equation}
    \omega_x(b,e|y,y') \;:=\; P^{y'}_\star(b,e|x,y).
    \label{eq:app_omega}
\end{equation}
Now, we construct three systems:

\emph{Bob's system.} Fiducial measurements labelled by $y\in\mathcal Y$ with outcomes $b$. Its state space $\Omega_B$ is the set of tables $r(b|y)$ that are nonnegative, normalized for each $y$, and satisfy every measurement OPEQ, $\sum_{b,y}\beta_{by}r(b|y)=\sum_{b,y}\beta'_{by}r(b|y)$. Bob's effects are the fiducial effects $r\mapsto r(b|y)$. Because the OPEQs hold on all of $\Omega_B$, each measurement OPEQ is an identity between effects of this system, exactly as required in Sec.~\ref{sec:contextuality}.

\emph{Eve's system.} Fiducial measurements labelled by $y'\in\mathcal Y$ with outcomes $e$, and state space $\Omega_E$ the set of all nonnegative normalized tables $s(e|y')$. This is a generalized bit in the sense of Ref.~\cite{barrett2007information} with $|\mathcal Y|$ fiducial measurements.

\emph{Alice's system.} Consider an independent vector per preparation, the standard basis $\{e_x\}$ of $\mathbb R^{\mathcal X}$, and identify precisely those combinations that the preparation OPEQs declare equal. Concretely, let $W\subseteq\mathbb R^{\mathcal X}$ be the span of the vectors $\alpha-\alpha'$ over all preparation OPEQs, define $V_A=\mathbb R^{\mathcal X}/W$, and let $\omega^A_x=[e_x]$ be the equivalence class of $e_x$. The state space $\Omega_A$ is the convex hull of $\{\omega^A_x\}$. By construction, a linear combination $\sum_x c_x\,\omega^A_x$ vanishes if and only if $\sum_x c_x e_x\in W$, so the linear relations among Alice's states are exactly the declared OPEQs and their linear consequences. 

We now verify that $\omega_x\in\Omega_B\otimes_{\max}\Omega_E$. No-signaling from $y'$ to $b$ is Eq.~\eqref{eq:app_F_marg}: $\sum_e\omega_x(b,e|y,y')=p_{\mathrm{obs}}(b|x,y)$ for every $y'$. No-signaling from $y$ to $e$ is Eq.~\eqref{eq:app_F_NS} applied to each slice. Conditioning on Eve's outcome $e$ of readout $y'$ gives Bob's conditional table $r^{x}_{e,y'}(b|y)\propto P^{y'}_\star(b,e|x,y)$, which is well defined because the marginal $\sum_b P^{y'}_\star(b,e|x,y)$ is independent of $y$ by Eq.~\eqref{eq:app_F_NS}, and which satisfies the measurement OPEQs by Eq.~\eqref{eq:app_F_meas} and hence, lies in $\Omega_B$. Conditioning on Bob's outcome gives a normalized table $s^{x}_{b,y}(e|y')\propto P^{y'}_\star(b,e|x,y)$ for Eve, which lies in $\Omega_E$ trivially. This holds for every $x$, so each $\omega_x$ is a valid joint state.

Finally, define Eve's intervention as the linear map $\mathcal T:V_A\to V_B\otimes V_E$ by $\mathcal T(\omega^A_x)=\omega_x$. Two properties make $\mathcal T$ a legitimate transformation. First, it is well defined on $V_A$: since $\omega^A_x$ is only defined up to elements of $W$, we need $\mathcal T$ to vanish on $W$, i.e.\ $\sum_x(\alpha_x-\alpha'_x)\,\omega_x=0$ for every preparation OPEQ, which is precisely Eq.~\eqref{eq:app_F_prep} applied to each slice. Second, it maps states to states: any state of Alice's system is a mixture $\sum_x p_x\,\omega^A_x$, whose image $\sum_x p_x\,\omega_x$ is a mixture of valid joint states and hence, a valid joint state itself because the joint state space is convex. Operationally, $\mathcal T$ is the statement that Eve takes the system Alice sent and outputs two systems, one forwarded to Bob and one she retains, jointly in the state $\omega_x$. It depends neither on $y$ nor on $y'$. These enter only afterwards, as the fiducial measurements that Bob and Eve respectively choose to perform on their halves. Every optimal slice $P^{y'}_\star$ is thereby recovered as one measurement, $y'$, on Eve's half of one and the same state, which is precisely what shows that the slices coexist within a single theory. 

The theory in which Alice's system, Bob's system, Eve's system and the transformation $\mathcal T$ are all allowed operations is a GPT in the sense of Ref.~\cite{barrett2007information}, and in it the protocol with Alice preparing $\omega^A_x$, Eve applying $\mathcal T$, Bob performing fiducial measurement $y$ and Eve performing fiducial readout $y'=y$ produces exactly the behaviour $p(b,e|x,y,y')=P^{y'}_\star(b,e|x,y)$, so the guessing probability of Proposition~\ref{prop:app_separation} is achieved.

It is important to consider two remarks. First, the construction does not claim that this GPT is the theory describing the actual devices. It shows that the maximally permissive adversarial model coincides with the set of attacks available in \emph{some} GPT respecting the OPEQs, so that the LP value is tight over that class rather than a loose relaxation. Second, restricting Eve to a more structured theory, quantum mechanics for instance, can only remove elements from $\Omega_B\otimes_{\max}\Omega_E$ and from the allowed transformations, and therefore can only decrease her optimal guessing probability.

\section{Hidden Eve-correlated randomness}
\label{app:hidden_randomness}

In this appendix, we illustrate why the operational equivalences used in our security analysis must hold before averaging over any uncharacterized variable correlated with Eve. We do so by showing an explicit attack exploiting this loophole in the CHSH prepare-and-measure protocol of the main text.

In this protocol, we have four preparations $P_x$, $x\in\{0,1,2,3\}$, satisfying
\begin{equation}
    \frac{1}{2}(P_0+P_1)\simeq\frac{1}{2}(P_2+P_3).
    \label{eq:B1}
\end{equation}
We identify the label of the preparation with its binary representation, i.e., we use two bits to indicate the preparation, $x=(x_0,x_1)$.

Now, suppose that Alice's preparation device contains an uncharacterized, uniformly random pair of bits $\lambda=(\lambda_0,\lambda_1)$, which is correlated with Eve. In addition to the intended system $Q$, prepared in the state $\rho^Q_x$, the device emits an additional classical degree of freedom $S$ encoding $(x_0\oplus\lambda_0,\,x_1\oplus\lambda_1)$. The physical preparation leaving Alice's laboratory is therefore
\begin{equation}
    P^{(\lambda)}_x=\rho^Q_x\otimes|x_0\oplus\lambda_0,\,x_1\oplus\lambda_1\rangle\langle x_0\oplus\lambda_0,\,x_1\oplus\lambda_1|^{S}.
    \label{eq:B2}
\end{equation}
Both bits of the preparation label are thus one-time padded by $\lambda$.

If Alice's characterization does not resolve $\lambda$, averaging over the hidden bits preserves the operational equivalence Eq.~\eqref{eq:B1}, and it does so for any test on the full carrier, including those acting jointly on $Q$ and $S$. Indeed, since $\lambda$ is uniform, the averaged state of $S$ is maximally mixed and independent of $x$,
\begin{equation}
    \sum_\lambda q(\lambda)\,P^{(\lambda)}_x=\rho^Q_x\otimes\frac{\mathds{1}^S}{4},
    \label{eq:B3a}
\end{equation}
so that Eq.~\eqref{eq:B1} on the carrier reduces to the equivalence of the intended components,
\begin{equation}
    \frac{1}{2}(\rho^Q_0+\rho^Q_1)=\frac{1}{2}(\rho^Q_2+\rho^Q_3).
    \label{eq:B3}
\end{equation}
Bob may furthermore act only on $Q$, so the observed probabilities $p(b|x,y)$ remain exactly those of the intended prepare-and-measure CHSH implementation.

For Eve, however, the two mixtures are not operationally equivalent. Since she knows $\lambda$, reading $S$ allows her to reconstruct the complete preparation label,
\begin{equation}
    x_0=S_0\oplus\lambda_0,\qquad x_1=S_1\oplus\lambda_1.
    \label{eq:B4}
\end{equation}
Once Bob announces his setting $y$, she can therefore compute Alice's raw key value directly,
\begin{equation}
    K=\kappa(x,y),
    \label{eq:B5}
\end{equation}
and hence
\begin{equation}
    P_{\mathrm{guess}}(K|E,Y)=1.
    \label{eq:B6}
\end{equation}
Thus, an operational equivalence that holds only after averaging over an Eve-correlated variable does not constrain Eve's side information, even though Alice and Bob observe the expected prepare-and-measure statistics, and even if their characterization of the carrier is otherwise tomographically complete.

More generally,
\begin{equation}
    \sum_\lambda q(\lambda)\sum_x\alpha_x P^{(\lambda)}_x
    \simeq
    \sum_\lambda q(\lambda)\sum_x\alpha'_x P^{(\lambda)}_x
    \label{eq:B7}
\end{equation}
is insufficient when $\lambda$ is available to Eve but hidden to Alice. Instead, the equivalence must hold conditioned on the variable,
\begin{equation}
    \sum_x\alpha_x P^{(\lambda)}_x
    \simeq
    \sum_x\alpha'_x P^{(\lambda)}_x
    \qquad\forall\,\lambda.
    \label{eq:B8}
\end{equation}
In the example above Eq.~\eqref{eq:B7} holds while Eq.~\eqref{eq:B8} fails, since for fixed $\lambda$ the two mixtures carry different values of $S_0$.

Importantly, this does not exclude ordinary shared or public randomness. Such variables may be included in the operational description and the OPEQs verified conditioned on their values. What is excluded is an uncharacterized Eve-correlated degree of freedom whose distinguishability is hidden only by averaging.

\section{Contextuality as a necessary condition for key generation}
\label{app:ncontextual_zero_key}

In this Appendix we prove Theorem~\ref{theorem} from section~\ref{sec:theorem} which states that contextuality is a necessary condition for obtaining a positive asymptotic key rate within the adversarial model considered in this work. More precisely, if the observed prepare-and-measure experiment, including its operational equivalences, admits a generalized noncontextual ontological model, then there exists an admissible individual attack for Eve under which no secret key can be distilled.

\begin{proof}
Suppose that the observed experiment admits a generalized noncontextual ontological model. Then, there exists an ontic space $\Lambda$, preparation distributions $\mu(\lambda|P_x)$, and measurement response functions $\xi(b|M_y,\lambda)$ such that the observed statistics are reproduced,
\begin{equation}
    p_{\mathrm{obs}}(b|x,y)
    =
    \sum_{\lambda\in\Lambda}
    \mu(\lambda|P_x)\,
    \xi(b|M_y,\lambda),
    \label{eq:app_ncom_model}
\end{equation}
and all specified OPEQs hold at the ontological level,
\begin{align}
    \sum_x \alpha_x\mu(\lambda|P_x)
    &=
    \sum_x \alpha'_x\mu(\lambda|P_x)
    \qquad\forall\,\lambda,
    \label{eq:app_ncom_prep}\\
    \sum_{b,y}\beta_{by}\xi(b|M_y,\lambda)
    &=
    \sum_{b,y}\beta'_{by}\xi(b|M_y,\lambda)
    \qquad\forall\,\lambda.
    \label{eq:app_ncom_meas}
\end{align}
For a finite prepare-and-measure scenario with finitely many OPEQs, a noncontextual model exists if and only if one exists with finitely many ontic states~\cite{schmid2018all}. We therefore take $|\Lambda|<\infty$ without loss of generality.

Now, we construct an explicit adversarial extension in which Eve's outcome is the ontic state itself,
\begin{equation}
    p(b,e{=}\lambda|x,y)
    =
    \mu(\lambda|P_x)\,
    \xi(b|M_y,\lambda).
    \label{eq:app_lambda_extension}
\end{equation}
This construction does not require $\lambda$ to be an operationally accessible variable in the physical realization of the experiment. Eq.~\eqref{eq:app_lambda_extension} simply defines a behaviour, and the security analysis must hold for every behaviour compatible with the observed statistics, the OPEQs and the causal structure of the attack model. Hence, it suffices to show that this one is admissible. Since the extension does not depend on Eve's readout choice, it is valid both in the fixed-readout and in the delayed-readout models of Appendix~\ref{app:fixed_setting}.

Positivity and normalization are straightforward. Marginalizing over Eve reproduces the observed statistics (by using Eqs.~\ref{eq:app_ncom_model} and \ref{eq:app_lambda_extension}),

\begin{equation}
    \sum_{\lambda}
    p(b,\lambda|x,y)=p_{\mathrm{obs}}(b|x,y),
\end{equation}
and the no-signaling condition from Bob's setting to Eve's outcome holds, because $\sum_b p(b,\lambda|x,y)=\mu(\lambda|P_x)$ is independent of $y$. Eq.~\eqref{eq:app_ncom_prep} implies
\begin{align}
    \sum_x\alpha_x p(b,\lambda|x,y)
    &=
    \xi(b|M_y,\lambda)
    \sum_x\alpha_x\mu(\lambda|P_x)
    \nonumber\\
    &=
    \xi(b|M_y,\lambda)
    \sum_x\alpha'_x\mu(\lambda|P_x)
    \nonumber\\
    &=
    \sum_x\alpha'_x p(b,\lambda|x,y),
\end{align}
so every preparation OPEQ is satisfied for each value of Eve's outcome. Likewise, Eq.~\eqref{eq:app_ncom_meas} gives
\begin{align}
    \sum_{b,y}\beta_{by}p(b,\lambda|x,y)
    &=
    \mu(\lambda|P_x)
    \sum_{b,y}\beta_{by}\xi(b|M_y,\lambda)
    \nonumber\\
    &=
    \mu(\lambda|P_x)
    \sum_{b,y}\beta'_{by}\xi(b|M_y,\lambda)
    \nonumber\\
    &=
    \sum_{b,y}\beta'_{by}p(b,\lambda|x,y),
\end{align}
so every measurement OPEQ is satisfied as well. The extension Eq.~\eqref{eq:app_lambda_extension} is therefore admissible in our adversarial model.

Since $\lambda$ determines Bob's response distribution, Eve can in particular simulate Bob's outcome distribution, and hence guess the raw key at least as well as Bob does. More formally, conditioned on $\lambda$ and $y$, Bob's outcome is generated solely by $\xi(b|M_y,\lambda)$ and is independent of Alice's preparation label; since the settings are chosen independently of everything else, the joint distribution factorizes as $q(x)\,q(y)\,\mu(\lambda|P_x)\,\xi(b|M_y,\lambda)$, and hence
\begin{equation}
    I(X;B|\Lambda,Y)=0 .
    \label{eq:app_conditional_independence}
\end{equation}
Because the attack is individual, the rounds form an i.i.d.\ source in which Alice holds $(X,Y)$, Bob holds $(B,Y)$ and Eve holds $(\Lambda,Y)$. For such a source, the secret-key rate $S(X;B\,\|\,\Lambda,Y)$, that is, the optimal rate at which Alice and Bob can distill key by authenticated public discussion, one-way or two-way, against an adversary holding $(\Lambda,Y)$, is upper bounded by the conditional mutual information~\cite{maurer1993secret},
\begin{equation}
    S\big(X;B\,\|\,\Lambda,Y\big)
    \leq
    I(X;B|\Lambda,Y)
    =0 .
\end{equation}

We conclude that any prepare-and-measure experiment admitting a generalized noncontextual ontological model with respect to the specified OPEQs admits an adversarial extension with zero secret-key capacity. Within the individual-attack adversarial model considered here, contextuality is therefore a necessary condition for a positive secret-key rate.
\end{proof}

\section{Quantum versus post-quantum adversaries: the $(3,2)$-PORAC.}
\label{app:PORAC}
In this Appendix, we have an explicit explanation of the $(3,2)$-PORAC protocol used in Ref.~\cite{chaturvedi2021characterising}, where they study contextuality-based QKD with a quantum adversary. This example demonstrates that the adversarial model considered here is strictly stronger than restricting Eve to quantum mechanics as we obtain no positive key rate.

In this protocol, Alice chooses a uniformly random three-bit string $x=(x_0,x_1,x_2)\in\{0,1\}^3$ and prepares the system accordingly in one of eight possible preparations while Bob has three possible settings $y\in\{0,1,2\}$. Then, his goal is to recover the corresponding bit $x_y$. Parity obliviousness requires the communicated system to contain no information about the parity of any subset of two or more of Alice's bits, giving rise to preparation operational equivalences. Concretely,
\begin{equation}
\begin{aligned}
    \frac{1}{4} (P_{000}+P_{011}+P_{101}+P_{110}) \simeq \\
    \frac{1}{4} (P_{001}+P_{010}+P_{100}+P_{111}),
    \label{eq:porac_opeq_example}
\end{aligned}
\end{equation}
\begin{equation}
    \frac{1}{2} (P_{00x_2}+P_{11x_2}) \simeq
    \frac{1}{2} (P_{01x_2}+P_{10x_2})  \qquad \forall x_2,
    \label{eq:porac_opeq_example2}
\end{equation}

with analogous equivalences for the other pairs of bits. A quantum realization satisfying these operational equivalences and achieving the maximum quantum bound for this task is depicted in Fig.~\ref{fig:porac3}. The operational equivalences for the measurements in this case are 
\begin{equation}
    \frac{1}{2}\left(M^0_0+M^1_0\right) \simeq
    \frac{1}{2}\left(M^0_1+M^1_1\right) \simeq
    \frac{1}{2}\left(M^0_2+M^1_2\right).
    \label{eq:porac_meas_opeq}
\end{equation}
which are just the trivial ones,  expressing nothing beyond the normalization of each measurement.

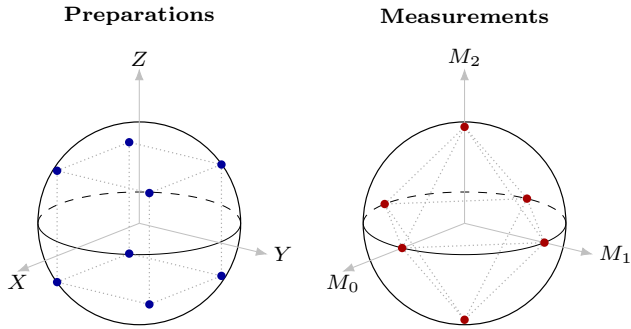
\begin{figure}[h!]
    \centering
    \usetikzlibrary{shapes.gates.logic.US, positioning}

    \definecolor{boxblue}{RGB}{125, 206, 245}

    \resizebox{1\columnwidth}{!}{%
    \begin{tikzpicture}[
  >={Latex[length=1.3mm]},
  x={(-0.7388cm,-0.2922cm)},
  y={( 0.9456cm,-0.2284cm)},
  z={( 0cm     , 1.1413cm)},
  st/.style={circle,fill=blue!60!black,inner sep=1.0pt},
  ef/.style={circle,fill=red!65!black,inner sep=1.0pt},
  ax/.style={gray!60,very thin,->},
  wire/.style={gray!80,densely dotted},
  eqb/.style={black,very thin,dashed},
  eqf/.style={black,very thin},
  lbl/.style={font=\scriptsize,text=black},
  tny/.style={font=\tiny},
  ttl/.style={font=\scriptsize\bfseries},
]
 
\def\q{0.5774}   %
\def\D{3.85}     %
\def\Ax{1.95}   %
\def\Ay{1.60}   %
\def\Az{1.60}   %
 
\begin{scope}
  \draw[eqb] plot[domain=128:308,samples=40,variable=\t] ({cos(\t)},{sin(\t)},0);
  \draw[eqf] plot[domain=-52:128,samples=40,variable=\t] ({cos(\t)},{sin(\t)},0);
  \draw[thin] (0,0,0) circle (1.2cm);
 
  \draw[ax] (0,0,0) -- (\Ax,0,0) node[lbl,below=-2pt]{$X$};
  \draw[ax] (0,0,0) -- (0,\Ay,0) node[lbl,right=-1pt]{$Y$};
  \draw[ax] (0,0,0) -- (0,0,\Az) node[lbl,above=-2pt]{$Z$};
 
  \draw[wire] (\q,\q,\q) -- (\q,\q,-\q) -- (\q,-\q,-\q) -- (\q,-\q,\q) -- cycle;
  \draw[wire] (-\q,\q,\q) -- (-\q,\q,-\q) -- (-\q,-\q,-\q) -- (-\q,-\q,\q) -- cycle;
  \draw[wire] (\q,\q,\q) -- (-\q,\q,\q);
  \draw[wire] (\q,\q,-\q) -- (-\q,\q,-\q);
  \draw[wire] (\q,-\q,\q) -- (-\q,-\q,\q);
  \draw[wire] (\q,-\q,-\q) -- (-\q,-\q,-\q);
 
  \node[st] at (\q,\q,\q) {};      \node[st] at (\q,\q,-\q) {};
  \node[st] at (\q,-\q,\q) {};     \node[st] at (\q,-\q,-\q) {};
  \node[st] at (-\q,\q,\q) {};     \node[st] at (-\q,\q,-\q) {};
  \node[st] at (-\q,-\q,\q) {};    \node[st] at (-\q,-\q,-\q) {};
 
  \node[ttl] at (0,0,2.15) {Preparations};
\end{scope}
 
\begin{scope}[xshift=\D cm]
  \draw[eqb] plot[domain=128:308,samples=40,variable=\t] ({cos(\t)},{sin(\t)},0);
  \draw[eqf] plot[domain=-52:128,samples=40,variable=\t] ({cos(\t)},{sin(\t)},0);
  \draw[thin] (0,0,0) circle (1.2cm);
 
  \draw[ax] (0,0,0) -- (\Ax,0,0) node[lbl,below=-2pt]{$M_{0}$};
  \draw[ax] (0,0,0) -- (0,\Ay,0) node[lbl,right=-1pt]{$M_{1}$};
  \draw[ax] (0,0,0) -- (0,0,\Az) node[lbl,above=-2pt]{$M_{2}$};
 
  \draw[wire] (1,0,0) -- (0,1,0) -- (-1,0,0) -- (0,-1,0) -- cycle;
  \draw[wire] (1,0,0) -- (0,0,1);   \draw[wire] (0,1,0) -- (0,0,1);
  \draw[wire] (-1,0,0) -- (0,0,1);  \draw[wire] (0,-1,0) -- (0,0,1);
  \draw[wire] (1,0,0) -- (0,0,-1);  \draw[wire] (0,1,0) -- (0,0,-1);
  \draw[wire] (-1,0,0) -- (0,0,-1); \draw[wire] (0,-1,0) -- (0,0,-1);
 
  \node[ef] at (1,0,0) {};   \node[ef] at (-1,0,0) {};
  \node[ef] at (0,1,0) {};   \node[ef] at (0,-1,0) {};
  \node[ef] at (0,0,1) {};   \node[ef] at (0,0,-1) {};
 
  \node[ttl] at (0,0,2.15) {Measurements};
\end{scope}
 
\end{tikzpicture}
    }

    \caption{The $3\!\to\!1$ parity-oblivious random access code. The eight preparations $P_{x_{0}x_{1}x_{2}}$ sit at the vertices of a cube inscribed in the Bloch sphere, with Bloch vector $\bigl((-1)^{x_{0}},(-1)^{x_{1}},(-1)^{x_{2}}\bigr)/\sqrt{3}$. The three measurements $M_{y}$ are the $X$, $Y$ and $Z$ bases, whose six effects form an octahedron. %
}
    \label{fig:porac3}
\end{figure}

This communication task naturally defines a key-distribution protocol. After Bob announces $y$, Alice retains the correct bit, i.e. $\kappa(x,y)=x_y$, while Bob uses his outcome $b$ as his estimate of the same raw key bit. This is precisely the key assignment considered in
Ref.~\cite{chaturvedi2021characterising}. There, restricting the adversarial extension to quantum theory leads to a positive secret-key rate for sufficiently high PORAC success probability. Concretely, for the maximum quantum value, they obtain approximately $r \approx 0.255$. In contrast, when Eve has access to arbitrary post-quantum extensions constrained only by the operational equivalences, our optimization gives no positive key rate. This example explicitly demonstrates the additional power granted to the adversary when the assumption of quantum mechanics is removed.

\section{Table of lower bounds on the key rates for different protocols}
\label{app:keyrun_tables}

In this appendix we collect in Table~\ref{tab:keyrun_optimum_alice} the best asymptotic lower bounds on the key rates we have found for a range of prepare-and-measure contextuality scenarios, treated as key-distribution protocols with the master key held by Alice. All rates are the LP reverse-Fano bound of Sec.~\ref{sec:adversarial_math} (note that the guessing probability is at least $1/2$ in every case, so Eq.\eqref{eq:reverse_fano} coincides with the general reverse-Fano bound and is tight even for the non-binary alphabets) against an individual attack by an adversary constrained only by the operational equivalences of the scenario. The reported values are the rate per \emph{key-generating} round rather than per experimental round. We report this value because, as in the efficient variant of the CHSH QKD protocol of Ref.~\cite{acin2006efficient}, Alice and Bob may bias the sampling of settings towards those used for key generation, so that in the asymptotic limit the rate per experimental round approaches the rate per key-generating round while the remaining rounds still suffice to estimate the operational statistics. For each protocol we report the best rate we have found, together with the choice of preparations $x$ assigned to key generation for each measurement setting $y$ that attains it. These are the best values found by our search over preparation subsets. We do not claim that they are optimal, and in particular, for the larger scenarios, they should be read as lower bounds on the achievable rate. We do not describe the protocols here, that is, which preparations and measurements the labels $x$ and $y$ refer to in each scenario; these definitions, together with the code that produces every entry of the table, are available in the public repository~\cite{ContextualityQKDGitHub}. Original references for each scenario are given in the first column of the table.

Two remarks help to read the numbers in Table~\ref{tab:keyrun_optimum_alice}. First, the entry for the prepare-and-measure CHSH protocol reproduces a known result: for the ideal quantum statistics our linear program returns the same noiseless lower bounds on the key rate that Ref.~\cite{acin2006bell} obtained for a no-signaling eavesdropper performing individual attacks on the CHSH protocol. This is expected, since in that scenario the operational constraints are exactly the no-signaling constraints which are mirrored by the OPEQs in the prepare-and-measure scenario. Second, the rates of $0.75$ and $1$ for the Cabello~\cite{cabello1996bell} and Peres~\cite{peres1991two} scenarios should not be read as Eve being fully ignorant. In both cases the preparations used for key generation within a given setting form a set of four, so the key symbol carries up to 2 bits. A rate of $1$ bit corresponds, through the reverse Fano inequality, to a guessing probability of $1/2$ for a four-valued symbol, that is, Eve retains at most one of the two bits. Random guessing would instead give a guessing probability of $1/4$ and a rate of $2$ bits per key-generating round. What is notable about these examples is therefore not the absolute value but the adversary against which it holds, the no-signaling literature reaches one secret bit per round only in the limit of infinitely many settings of the chained Bell inequality~\cite{acin2006bell}, whereas here it is attained with a finite number of settings against an adversary constrained only by the operational equivalences of the scenario. Note that these are noiseless values, their robustness to noise is left for future work. 

\begin{table}[h]
\scriptsize
\setlength{\tabcolsep}{3pt}
\begin{tabular}{lcrl}
\hline\hline
Protocol & $X\times Y$ & Rate & Key preparations per setting \\
\hline
\begin{tabular}[t]{@{}l@{}}Prepare-and-\\measure CHSH \end{tabular} & $4\times2$ & 0.1062 & \begin{tabular}[t]{@{}l@{}} $y_{0}{:}\lbrace 0,1,2,3\rbrace$; \\ $y_{1}{:}\lbrace 0,1,2,3\rbrace$ \end{tabular} \\
aligned key CHSH & $6\times2$ & 0.5857 & $y_{0}{:}\lbrace 4,5\rbrace$ \\
Hexagon \cite{mazurek2016experimental} & $6\times3$ & 0.5000 & \begin{tabular}[t]{@{}l@{}} $y_{0}{:}\lbrace 0,3\rbrace$; $y_{1}{:}\lbrace 2,5\rbrace$; \\ $y_{2}{:}\lbrace 1,4\rbrace$ \end{tabular} \\
$(3,2)$-PORAC & $8\times3$ & 0 & none \\
\begin{tabular}[t]{@{}l@{}}Icosahedron /\\dodecahedron\\\cite[Fig.~3b]{zhang2026entanglement} \end{tabular} & $20\times6$ & 0.3325 & \begin{tabular}[t]{@{}l@{}} $y_{0}{:}\lbrace 0,4,9,11\rbrace$; \\ $y_{2}{:}\lbrace 0,2,17,19\rbrace$; \\ $y_{4}{:}\lbrace 0,1,13,15\rbrace$ \end{tabular} \\
Cabello 18-ray~\cite{cabello1996bell} & $18\times9$ & 0.7500 & \begin{tabular}[t]{@{}l@{}} $y_{0}{:}\lbrace 0,1,10,11\rbrace$; \\ $y_{1}{:}\lbrace 0,2,12,13\rbrace$; \\ $y_{2}{:}\lbrace 1,2,15,16\rbrace$; \\ $y_{3}{:}\lbrace 3,4,13,14\rbrace$; \\ $y_{4}{:}\lbrace 3,5,15,17\rbrace$; \\ $y_{5}{:}\lbrace 4,5,9,10\rbrace$; \\ $y_{6}{:}\lbrace 6,7,9,11\rbrace$; \\ $y_{7}{:}\lbrace 6,8,16,17\rbrace$; \\ $y_{8}{:}\lbrace 7,8,12,14\rbrace$ \end{tabular} \\
Peres 24-ray~\cite{peres1991two} & $24\times6$ & 1.0000 & \begin{tabular}[t]{@{}l@{}} $y_{0}{:}\lbrace 0,1,2,3\rbrace$; \\ $y_{1}{:}\lbrace 4,5,6,7\rbrace$; \\ $y_{2}{:}\lbrace 8,9,10,11\rbrace$; \\ $y_{3}{:}\lbrace 12,13,14,15\rbrace$; \\ $y_{4}{:}\lbrace 16,17,18,19\rbrace$; \\ $y_{5}{:}\lbrace 20,21,22,23\rbrace$ \end{tabular} \\
\hline\hline
\end{tabular}
\caption{Best asymptotic lower bounds on the key rate we have found, in bits per key-generating round, together with the preparations $x$ assigned to key generation for each measurement setting $y$ that attain it. These values result from a search over preparation subsets and are not claimed to be optimal. ``none'' means that no setting was found to yield a positive rate.}
\label{tab:keyrun_optimum_alice}
\end{table}

\end{document}